\documentclass[journal]{IEEEtran}
\usepackage{amsmath,amssymb,amsthm,bm}
\usepackage{graphicx}
\usepackage{cuted}
\usepackage{capt-of}
\usepackage{placeins}
\usepackage[hidelinks]{hyperref}
\hypersetup{
  pdftitle={Fundamental Limits of Adaptive Beamforming Under Finite Training},
  pdfauthor={Zhiyong Cheng, Shengyao Chen, and Di Song},
  pdfsubject={Sharp first-order local minimax limits of adaptive beamforming from finite training data},
  pdfkeywords={adaptive beamforming, finite training, fundamental limits, local asymptotic minimax theory, MVDR, Fisher information},
  hypertexnames=false
}
\newtheorem{theorem}{Theorem}
\newtheorem{lemma}{Lemma}
\newtheorem{proposition}{Proposition}
\newtheorem{corollary}{Corollary}
\newtheorem{assumption}{Assumption}
\theoremstyle{remark}
\newtheorem{remark}{Remark}
\newtheorem*{informalresult}{Informal main result}

\DeclareMathOperator{\tr}{tr}
\newcommand{\E}{\mathbb{E}}
\newcommand{\R}{\mathbb{R}}
\newcommand{\Cx}{\mathbb{C}}
\newcommand{\ba}{\bar{\mathbf a}}
\newcommand{\w}{\mathbf w}
\newcommand{\x}{\mathbf x}
\newcommand{\q}{\mathbf q}
\newcommand{\wopt}{\w_\star}
\newcommand{\Rb}{\mathbf R}
\newcommand{\Mb}{\mathbf M}
\newcommand{\Jb}{\mathbf J}
\newcommand{\Ib}{\mathbf I}
\newcommand{\Sb}{\mathbf S}
\newcommand{\fb}{\mathbf f}
\newcommand{\Ab}{\mathbf A}
\newcommand{\Bb}{\mathbf B}
\newcommand{\Xb}{\mathbf X}
\newcommand{\Yb}{\mathbf Y}

\newcommand{\Hb}{\mathbf H}
\newcommand{\herm}{^{\!*}}
\newcommand{\inv}{^{-1}}
\newcommand{\at}{\tilde{\mathbf a}}
\newcommand{\qt}{\tilde{\mathbf q}}
\newcommand{\tha}{{\bm\theta}}
\newcommand{\thL}{{\bm\theta_L}}

\begin{document}

\title{Fundamental Limits of Adaptive Beamforming Under Finite Training}
\author{Zhiyong Cheng,~\IEEEmembership{Member,~IEEE},
Shengyao Chen,~\IEEEmembership{Member,~IEEE}, and
Di Song,~\IEEEmembership{Member,~IEEE}%
\thanks{Z. Cheng is with the School of Computer and Artificial Intelligence, Chaohu University (e-mail: chengzhiyong@ieee.org).}%
\thanks{S. Chen is with the School of Electronic and Optical Engineering, Nanjing University of Science and Technology (e-mail: chenshengyao@njust.edu.cn). (Corresponding author: Shengyao Chen.)}%
\thanks{D. Song is with the School of Physics and Electronic Engineering, Nanyang Normal University (e-mail: dsong1\_13@njust.edu.cn).}%
}
\maketitle

\begin{abstract}
Finite training reduces the output signal-to-interference-plus-noise ratio (SINR) of an adaptive beamformer, and a natural question is how much of this loss is unavoidable. This paper determines this question, providing a beamforming counterpart of the Cram\'er--Rao bound in spectral estimation. An exact identity expresses the SINR loss as a bounded function of the error in the clairvoyant minimum-variance distortionless-response (MVDR) weight. It yields a local asymptotic minimax lower bound over all measurable data-dependent beamforming rules, including biased and irregular rules. The first-order coefficient is $\tr(\Mb\Jb_{\rm eff}^{-1})$, where $\Jb_{\rm eff}$ describes the information in the training data and $\Mb$ measures the sensitivity of the output SINR. Matching constructions determine this coefficient in two complex-Gaussian models. For an $N$-sensor uniform linear array with $K$ distinct point interferers and $2K+1\le N$, a data-driven split one-step beamformer attains the coefficient $C_\theta\le K$ at every interior scene of a fixed compact regular parameter set. For unrestricted covariance matrices, sample matrix inversion (SMI) attains the coefficient $N-1$ through the classical Reed--Mallett--Brennan law. The difference quantifies the first-order value of finite-source structure. Geometric formulas and numerical results describe the dependence on interference power and array geometry, and the finite-sample departure near a weak-source boundary.
\end{abstract}


\section{Introduction}

\IEEEPARstart{A}{daptive} beamforming is used in radar, sonar, radio astronomy, and wireless communications to preserve a signal from a prescribed look direction while suppressing interference and noise. If the ensemble interference-plus-noise covariance is known, the minimum-variance distortionless-response (MVDR) beamformer maximizes the output signal-to-interference-plus-noise ratio (SINR) under the distortionless constraint \cite{VanTrees2002}. In practice the covariance is unknown, so the weight must be formed from $n$ signal-free training snapshots; the resulting weight is random, and its output SINR falls below the clairvoyant (known-covariance) value.

A basic question is how large this loss must be. In line spectral estimation and direction-of-arrival estimation, the analogous question is answered by the Cram\'er--Rao bound (CRB), which relates Fisher information to estimation error and has long served as an algorithm-independent benchmark for estimator efficiency and system design \cite{StoicaNehorai1989,StoicaNehorai1990}. A beamforming counterpart must be stated in terms of the output-SINR loss, must account for both the information in the training data and the effect of parameter errors on the optimum weight, and must apply to the whole class of data-dependent beamforming rules.

The classical finite-training benchmark is the Reed--Mallett--Brennan (RMB) law for sample matrix inversion (SMI) \cite{ReedMallettBrennan1974}: under independent complex-Gaussian training, the normalized output SINR of SMI has an exact beta distribution, and its mean loss has first-order coefficient $N-1$ for an array of $N$ sensors. This is the origin of the familiar rule of about two training snapshots per sensor for a $3$-decibel (dB) loss. But the RMB law only determines what SMI achieves. Here we ask a different question: what is the minimum output-SINR loss allowed by the training model itself, including its first-order coefficient, and which rule attains it?

The first step is an exact representation of the loss. Let $\Rb$ denote the interference-plus-noise covariance, let $\ba_0$ denote the known look-direction steering vector, i.e., the array's response to a unit signal from the look direction, and let $\wopt$ denote the clairvoyant MVDR weight. For a candidate weight $\hat\w$, denote the normalized output-SINR loss by $L(\hat\w;\Rb)$. The loss is unchanged by a nonzero scaling of $\hat\w$, so the weight can be normalized to have unit look gain. For a vector $\mathbf u$, define $\|\mathbf u\|_{\Rb}^{2}=\mathbf u\herm\Rb\mathbf u$, where $(\cdot)\herm$ denotes the Hermitian transpose. We show that
\begin{equation}\label{eq:intro-identity}
L(\hat\w;\Rb)
=\frac{\|\hat\w-\wopt\|_{\Rb}^{2}}
{P_\star+\|\hat\w-\wopt\|_{\Rb}^{2}},
\qquad
P_\star=\frac{1}{\ba_0\herm\Rb^{-1}\ba_0},
\end{equation}
where $P_\star$ is the clairvoyant MVDR output power. The identity in (\ref{eq:intro-identity}) shows that the SINR loss is a bounded function of the error in the optimum weight, with a quadratic leading term near the optimum. Finite-training beamforming can therefore be studied as estimation of the MVDR weight under the metric set by the beamforming task; no covariance estimator or beamforming algorithm needs to be fixed at this stage.

Let $\bm\vartheta$ denote the state of a regular finite-dimensional training model. Two matrices govern the first-order loss: the per-snapshot Fisher information $\Jb_{\vartheta}$, which measures how well nearby states can be distinguished, and the loss-curvature matrix $\Mb_{\vartheta}$, which measures how errors in those states change the output SINR. Together they give the coefficient
\begin{equation*}
C_{\vartheta}=\tr(\Mb_{\vartheta}\Jb_{\vartheta}^{-1}),
\end{equation*}
where $\tr(\cdot)$ denotes the matrix trace. In the regular fixed-dimensional regime, the $1/n$ rate is common to many models; the coefficient $C_{\vartheta}$ determines their different training requirements. A CRB calculation gives this form for unbiased or locally unbiased estimation, but it does not give a converse for arbitrary biased or irregular beamforming rules. We use the H\'ajek--Le Cam local asymptotic minimax theorem \cite{Hajek1972,vanderVaart1998} and (\ref{eq:intro-identity}) to obtain such a converse under the exact bounded output-SINR loss.

\begin{informalresult}
For every regular finite-dimensional training model formulated in Section~\ref{sec:framework}, $C_{\vartheta}/n$ is a first-order local minimax lower bound on the expected output-SINR loss over all measurable data-dependent beamforming rules. For complex-Gaussian training on an $N$-sensor uniform linear array (ULA) with $K$ distinct point interferers and $2K+1\le N$, a fully data-driven split one-step beamformer attains the bound at every interior scene of a fixed compact regular parameter set. The converse and achievability meet at
\begin{equation*}
\frac{C_\theta}{n},
\qquad
0<C_\theta=\tr(\Mb\Jb_L^{-1})
\le K,
\end{equation*}
where $\theta$ denotes the finite-source scene, $\Mb$ is its loss-curvature matrix, and $\Jb_L$ is the efficient Fisher information for the source directions and relative powers. For the unrestricted complex-Gaussian covariance model, the same converse and the SMI achievability meet at $(N-1)/n$. The SMI achievability follows from the classical RMB law.
\end{informalresult}

The attaining ULA construction has two stages: a least-squares fit on a small pilot block of snapshots gives a rough scene estimate, and one Fisher-scoring (Newton-type) correction on the remaining, independent snapshots makes it efficient. Because the covariance determines the scene uniquely and the parameter set is compact, the estimation error has bounded moments of every order; this is what converts the local expansion into a statement about the expected SINR loss. A quantitative lower bound on the source separation is needed only for the geometric estimates that are uniform in the array size, not for achievability at fixed $N$ and $K$.

The coefficient $C_\theta$ has a simple reading: $\Jb_L^{-1}$ describes the uncertainty that remains in the source directions and relative powers, and $\Mb$ keeps the part that changes the output SINR. For a small target mean loss $\epsilon>0$, an efficient rule therefore requires approximately $C_\theta/\epsilon$ training snapshots to first order.

A finer description comes from the eigenvalues of $\Jb_L^{-1/2}\allowbreak\Mb\allowbreak\Jb_L^{-1/2}$. They determine the limiting distribution of the scaled loss and identify which parameter directions matter to the beamforming task. We also derive a source-wise decomposition; at high interference-to-noise ratio (INR), each interferer contributes one half of a squared cosine between the projected look direction and the projected array-manifold tangent (the derivative of the steering vector with respect to the source direction).

The numerical results compare practical structured reconstructions with the finite-source coefficient, verify the sharp SMI coefficient, and trace the finite-sample onset of the local regime near a weak-source boundary, where an interferer's power is barely large enough for reliable localization.

\subsection{Relations to Prior Art}\label{sec:prior}

Sample-covariance beamformers under finite training have been analyzed from many angles: exact distributions and perturbation expansions are available for SMI and related processors \cite{Steinhardt1991,RaghunathReddy1992,Richmond1996}, while diagonal loading, covariance shrinkage, and random-matrix corrections reduce the instability of covariance inversion \cite{CoxZeskindOwen1987,ChenWieselEldarHero2010,MestreLagunas2006,YangMcKayCouillet2018}. These analyses give detailed finite-sample performance, but always for specified processors and estimator families; they do not determine the best first-order coefficient over the full rule class.

Another line of work exploits the low-dimensional structure of the interference. Dominant-mode and reduced-rank beamformers restrict adaptation to an estimated interference subspace \cite{AbrahamOwsley1990,WageBuck2014,WageBuck2015}; for one strong interferer, the mean dominant-mode-rejection loss depends on the number of training snapshots. Parametric reconstruction methods estimate the source directions and powers and then form the MVDR weight from the reconstructed model \cite{SantosZoltowskiRangaswamy2007}. Interference structure can evidently remove the array-size dependence of SMI; Section~\ref{sec:ula} identifies the model-wide coefficient behind this effect.

The MVDR weight has also been studied directly as the quantity to be estimated: Souloumiac derived a CRB for unbiased estimates of the minimum-variance weight and proposed a biased modification of the sample weight \cite{Souloumiac1996}, Ollila and Koivunen obtained the influence function and asymptotic covariance of scatter plug-in MVDR weights \cite{OllilaKoivunen2009}, and Besson developed Stein-type modifications of the sample MVDR filter \cite{Besson2024}. We share this view of the weight as the object of interest, but we work with the exact bounded output-SINR loss and a local minimax criterion over the full rule class. We also give attaining rules for the two models studied below.

Beyond the CRB, the Ziv--Zakai, Chazan--Zakai--Ziv, and Weiss--Weinstein bounds give nonlocal or prior-dependent lower bounds on parameter mean-square error \cite{ZivZakai1969,ChazanZakaiZiv1975,WeissWeinstein1985}, and they are the tools of choice when a local quadratic approximation is insufficient. The present problem differs in both action and loss: the action is a beamforming weight, and the loss is the exact bounded output-SINR degradation. It is Lemma~\ref{lem:transfer} that connects this task loss to the local asymptotic minimax bound.

Robust adaptive beamforming treats a different source of loss: it protects against steering-vector mismatch, covariance uncertainty, or distributional ambiguity through a prescribed uncertainty set or ambiguity model \cite{VorobyovGershmanLuo2003,GuLeshem2012,WangDaiLi2025,HuangHuangVorobyovLuo2026}. In the present problem the training model is fixed and the loss comes from finite sampling alone, so the resulting minimax bound concerns sampling uncertainty rather than a worst-case design over model uncertainty.

\subsection{Notation}\label{sec:notation}

Bold lowercase and uppercase letters denote vectors and matrices. The superscripts $(\cdot)\herm$ and $(\cdot)^\top$ denote the Hermitian transpose and transpose. We use $\|\cdot\|$ for the Euclidean or spectral norm and $\|\cdot\|_F$ for the Frobenius norm. For Hermitian matrices, $\Ab\succeq\Bb$ means that $\Ab-\Bb$ is positive semidefinite. The set of $N\times N$ Hermitian positive-definite matrices is $\mathbb H_{++}^{N}$. For $\Xb\succ0$, $\Xb^{1/2}$ is its Hermitian positive-definite square root, $\langle\mathbf u,\mathbf v\rangle_{\Xb}=\mathbf u\herm\Xb\mathbf v$, and $\|\mathbf u\|_{\Xb}^2=\mathbf u\herm\Xb\mathbf u$. For a nonzero vector $\mathbf v$ and a full-column-rank matrix $\Ab$, define $\mathbf P_{\mathbf v}^{\perp}=\Ib-\mathbf v\mathbf v\herm/\|\mathbf v\|^2$ and $\mathbf P_{\Ab}^{\perp}=\Ib-\Ab(\Ab\herm\Ab)^{-1}\Ab\herm$. The unit torus is $\mathbb T=\R/\mathbb Z$, with distance $|t|_{\mathbb T}=\min_{\ell\in\mathbb Z}|t-\ell|$. All derivatives are with respect to real coordinates; $Df$ denotes the Jacobian or Fr\'echet derivative and $\partial_a f$ a coordinate derivative.

The distributions $\mathcal{CN}(\bm\mu,\Rb)$ and $\mathcal N(\bm\mu,\mathbf\Sigma)$ are circular complex Gaussian and real Gaussian, respectively. We write $\E_{\vartheta}$ for expectation under the state $\bm\vartheta$, $\Rightarrow$ for convergence in distribution, and $a\wedge b=\min\{a,b\}$. The orders $O(\cdot)$ and $o(\cdot)$ are deterministic. We write $x\lesssim y$ when $x\le Cy$ for a regime-dependent constant independent of $N$ and $n$, and use $\gtrsim$ and $\asymp$ for the reverse and two-sided relations. In fixed-dimensional compact-set arguments, constants may depend on the chosen set, $N$, and $K$; dimension-uniform bounds state their dependencies explicitly.

\subsection{Organization}

The rest of this paper is organized as follows. Section~\ref{sec:framework} defines the training model, the output-SINR loss, and the local risk. Section~\ref{sec:general} gives the general local asymptotic minimax converse. Sections~\ref{sec:ula} and~\ref{sec:full} establish matching achievability for the finite-source ULA and unrestricted Gaussian models. Section~\ref{sec:visibility} interprets the finite-source coefficient. Section~\ref{sec:sim} presents the numerical results, and Sections~\ref{sec:discussion} and~\ref{sec:concl} close with the discussion and conclusion. All proofs are deferred to Appendices~\ref{app:excess}--\ref{app:smi}.
\section{Statistical model and SINR loss}\label{sec:framework}

This section defines the training model, the beamforming loss, and the local risk. Let $\mathcal V\subset\R^d$ be an open parameter set, where $d$ is the state dimension, and let $\bm\vartheta\in\mathcal V$ denote the unknown interference state. The one-snapshot training model is $\{P_{\bm\vartheta}:\bm\vartheta\in\mathcal V\}$. A signal-free snapshot $\x\sim P_{\bm\vartheta}$ has zero mean and positive-definite covariance $\Rb(\bm\vartheta)$. The $n$ snapshots are independent, and $P_{\bm\vartheta}^{\otimes n}$ denotes their joint distribution. We refer to the one-snapshot family as the training model and to the product family as the $n$-snapshot experiment. The look-direction steering vector $\ba_0\ne\bm0$ is known. A beamforming rule is any measurable mapping from the snapshots to a weight $\hat\w_n(\x_1,\ldots,\x_n)$.

The clairvoyant MVDR weight and its output power are
\begin{equation}\label{eq:general-mvdr}
\begin{aligned}
\psi(\bm\vartheta)=\wopt(\bm\vartheta)
&=\frac{\Rb(\bm\vartheta)^{-1}\ba_0}
{\ba_0\herm\Rb(\bm\vartheta)^{-1}\ba_0},\\
P_\star(\bm\vartheta)
&=\frac{1}{\ba_0\herm\Rb(\bm\vartheta)^{-1}\ba_0}.
\end{aligned}
\end{equation}
The normalized output-SINR loss is
\begin{equation}\label{eq:loss}
L(\hat\w_n;\bm\vartheta)
=1-\frac{|\hat\w_n\herm\ba_0|^2}
{(\hat\w_n\herm\Rb(\bm\vartheta)\hat\w_n)
(\ba_0\herm\Rb(\bm\vartheta)^{-1}\ba_0)}\in[0,1].
\end{equation}
We set $L(\bm0;\bm\vartheta)=1$, which is also the loss of every nonzero weight with zero look gain. We write $\rho(\hat\w_n;\bm\vartheta)=1-L(\hat\w_n;\bm\vartheta)$ for the normalized output-SINR ratio.

\begin{lemma}[Exact excess-loss identity]\label{lem:excess}
For every $\hat\w$ satisfying $\hat\w\herm\ba_0=1$,
\begin{equation}\label{eq:excess}
L(\hat\w;\bm\vartheta)
=\frac{E}{P_\star+E},
\qquad
E=\|\hat\w-\psi(\bm\vartheta)\|_{\Rb(\bm\vartheta)}^2,
\end{equation}
where $E$ is the excess output power above the clairvoyant MVDR value.
\end{lemma}

The proof is given in Appendix~\ref{app:excess}. Since the loss is invariant to a nonzero scaling of the weight, every weight with nonzero look gain can be normalized to satisfy the distortionless constraint. For a small error, (\ref{eq:excess}) gives $L=E/P_\star+o(E/P_\star)$: to leading order, the loss is the squared error in the optimum weight under the metric $\Rb(\bm\vartheta)/P_\star(\bm\vartheta)$.

We next describe how the training data and the beamforming task enter the loss. At a regular interior point, let $\Jb_{\vartheta}\in\R^{d\times d}$ denote the Fisher information in one snapshot. Define the task metric $\mathbf W_{\vartheta}\in\Cx^{N\times N}$, the Jacobian $\mathbf G_{\vartheta}\in\Cx^{N\times d}$ of the MVDR map, and the loss-curvature matrix $\Mb_{\vartheta}\in\R^{d\times d}$ as
\begin{equation}\label{eq:general-WM}
\begin{aligned}
\mathbf W_{\vartheta}&=\frac{\Rb(\bm\vartheta)}{P_\star(\bm\vartheta)},
&\mathbf G_{\vartheta}&=D\psi(\bm\vartheta),\\
\Mb_{\vartheta}&=\operatorname{Re}\!\left(
\mathbf G_{\vartheta}\herm\mathbf W_{\vartheta}\mathbf G_{\vartheta}
\right).
\end{aligned}
\end{equation}
Here $\Jb_{\vartheta}^{-1}$ describes the remaining local uncertainty of the state, while $\Mb_{\vartheta}$ weights each state direction by its effect on the output SINR; their trace pairing is
\begin{equation}\label{eq:general-C}
C_{\vartheta}=\tr(\Mb_{\vartheta}\Jb_{\vartheta}^{-1}).
\end{equation}
The value in (\ref{eq:general-C}) is unchanged by a smooth reparametrization of the state.

Some state coordinates may change the data distribution without changing the optimum weight; a Schur complement removes their effect from the Fisher information. Write $\bm\vartheta=(\bm\eta,\bm\nu)$, where $\bm\eta$ contains the coordinates that affect the MVDR weight and $\bm\nu$ contains the nuisance coordinates. If $D_{\nu}\psi=\bm0$, define
\begin{equation}\label{eq:general-Jeff}
\Jb_{\rm eff}
=\Jb_{\eta\eta}-\Jb_{\eta\nu}\Jb_{\nu\nu}^{-1}\Jb_{\nu\eta}.
\end{equation}
Then the block inverse identity gives
\begin{equation}\label{eq:general-profile}
C_{\vartheta}=\tr(\Mb_{\eta}\Jb_{\rm eff}^{-1}),
\qquad
\Mb_{\eta}=\operatorname{Re}\!\left((D_{\eta}\psi)\herm\mathbf W_{\vartheta}D_{\eta}\psi\right).
\end{equation}
Thus the nuisance coordinates are removed before the coefficient is evaluated. We use (\ref{eq:general-profile}) for the finite-source model.

At a fixed state, the infimum over all measurable rules is trivially zero, since one can choose the constant rule $\hat\w_n\equiv\psi(\bm\vartheta)$; a nontrivial algorithm-independent bound must therefore test the same rule over nearby states. Let $\bm h\in\R^d$ be a local perturbation and let $r>0$ be the neighborhood radius, and write $\bm\vartheta_{n,\bm h}=\bm\vartheta+\bm h/\sqrt n$ for the local alternatives, which remain distinguishable at the first-order scale. We define the local risk as
\begin{equation}\label{eq:local-risk}
\mathcal R_{n,r}(\bm\vartheta)
=\inf_{\hat\w_n}\sup_{\|\bm h\|\le r}
\E_{\bm\vartheta_{n,\bm h}}
\big[L(\hat\w_n;\bm\vartheta_{n,\bm h})\big],
\end{equation}
where the infimum is over all measurable beamforming rules and the supremum is over $\|\bm h\|\le r$. The expectation $\E_{\bm\vartheta_{n,\bm h}}$ is taken under $P_{\bm\vartheta_{n,\bm h}}^{\otimes n}$. Throughout the first-order analysis, the model dimension is fixed while the number of snapshots tends to infinity. In the ULA model, $N$ and $K$ are fixed. The next section derives the local minimax converse for this risk.

\section{A general local asymptotic minimax converse}\label{sec:general}

Fix an interior point $\bm\vartheta\in\mathcal V$, and assume that the one-snapshot model is quadratic-mean differentiable at $\bm\vartheta$, that $\Jb_{\vartheta}$ is nonsingular, and that $\Rb(\cdot)$ is continuously differentiable and positive definite nearby. These conditions imply local asymptotic normality of the $n$-snapshot experiment and differentiability of the MVDR map in (\ref{eq:general-mvdr}). Informally, quadratic-mean differentiability is a smoothness requirement on the square root of the one-snapshot density, and local asymptotic normality states that, at the $1/\sqrt n$ scale, the $n$-snapshot experiment is approximated by a Gaussian shift experiment with information matrix $\Jb_{\vartheta}$.

\begin{lemma}[Transfer to the exact SINR loss]\label{lem:transfer}
Let $\hat\w_n$ be any measurable beamforming rule. Normalize its output to unit look gain whenever the look gain is nonzero, and use $\ba_0/(\ba_0\herm\ba_0)$ on the zero-look-gain event. Denote the resulting action by $T_n$, and let
\begin{equation}\label{eq:transfer-loss}
\ell_n=\left\|T_n-\psi(\bm\vartheta_{n,\bm h})\right\|_{\Rb(\bm\vartheta_{n,\bm h})/P_\star(\bm\vartheta_{n,\bm h})}^2.
\end{equation}
For every truncation level $A>0$, the following inequality holds pointwise on the sample space:
\begin{equation}\label{eq:main-transfer}
nL(\hat\w_n;\bm\vartheta_{n,\bm h})
\ge
\frac{(n\ell_n)\wedge A}{1+A/n}.
\end{equation}
\end{lemma}

The proof is given in Appendix~\ref{app:general}. Lemma~\ref{lem:transfer} converts a minimax bound for the truncated quadratic loss into one for the exact bounded output-SINR loss, and it does so pointwise on the sample space, covering even rules whose output has zero look gain.

\begin{theorem}[General local asymptotic minimax SINR converse]\label{thm:general}
Under the preceding conditions,
\begin{equation}\label{eq:general-converse}
\lim_{r\to\infty}\liminf_{n\to\infty}
\inf_{\hat\w_n}\sup_{\|\bm h\|\le r}
 n\,\E_{\bm\vartheta_{n,\bm h}}
\big[L(\hat\w_n;\bm\vartheta_{n,\bm h})\big]
\ge C_{\vartheta},
\end{equation}
where $C_{\vartheta}$ is defined in (\ref{eq:general-C}).
\end{theorem}

Theorem~\ref{thm:general} applies to every measurable beamforming rule, biased and irregular rules included (no unbiasedness, consistency, or asymptotic normality is assumed). The training model enters through $\Jb_{\vartheta}$ and the beamforming task through $\Mb_{\vartheta}$; consequently, two training models with the same covariance map can have different beamforming limits if their observations carry different Fisher information. The proof, given in Appendix~\ref{app:general}, applies the H\'ajek--Le Cam theorem to the truncated quadratic loss in Lemma~\ref{lem:transfer} and then uses (\ref{eq:main-transfer}) to recover the exact output-SINR loss. Sections~\ref{sec:ula} and~\ref{sec:full} establish matching achievability for the finite-source ULA and unrestricted Gaussian models.

\section{The finite-source ULA limit}\label{sec:ula}

We now apply the general converse to a finite-source ULA and construct an attaining beamformer. The ULA has $N$ sensors at half-wavelength spacing. Its unit-norm steering vector at spatial frequency $\phi$ is $\ba(\phi)=N^{-1/2}(e^{i2\pi m\phi})_m\in\Cx^N$, where $m$ ranges over $\{-(N-1)/2,\ldots,(N-1)/2\}$. Half-integer indices are used when $N$ is even. The look direction is $\theta_0$, and $\ba_0=\ba(\theta_0)$. All source directions lie in a compact interval $\Phi_{\rm op}\subset(\theta_0-1/2,\theta_0+1/2)$.

For $1\le\ell\le n$, the signal-free training snapshots satisfy
\begin{equation}\label{eq:R}
\x_\ell\sim\mathcal{CN}(\bm0,\Rb(\tha)),
\qquad
\Rb(\tha)=\sigma^2\Ib+\sum_{j=1}^Kp_j\ba(\phi_j)\ba(\phi_j)\herm,
\end{equation}
where $K\ge1$ is the known and fixed number of interferers. The quantities $\phi_j$ and $p_j>0$ are the spatial frequency and power of source $j$, and $\sigma^2>0$ is the noise variance. The finite-source scene is $\tha=(\phi_1,\ldots,\phi_K,p_1,\ldots,p_K,\sigma^2)$. We use the ordering $\phi_1<\cdots<\phi_K$ and define
\begin{equation*}
\ba_j=\ba(\phi_j),
\qquad
\dot\ba_j=\frac{\partial\ba(\phi_j)}{\partial\phi_j},
\qquad
\q_j=(\Ib-\ba_j\ba_j\herm)\dot\ba_j.
\end{equation*}
For the centered ULA, $\ba_j\herm\dot\ba_j=0$, so $\q_j=\dot\ba_j$. Let $\Theta_{\rm reg}$ denote the scenes with distinct source directions, positive source powers and noise variance, and $2K+1\le N$. Appendix~\ref{app:reg} proves pointwise regularity on $\Theta_{\rm reg}$ and the compact-set bounds used below.

The MVDR weight is unchanged when all powers are multiplied by the same positive number. We therefore use the relative log interference-to-noise ratios (INRs) $\beta_j=\log(p_j/\sigma^2)$ and the global log scale $\gamma=\log\sigma^2$. Let $\bm\phi=(\phi_1,\ldots,\phi_K)$ and $\bm\beta=(\beta_1,\ldots,\beta_K)$. In the coordinates $\tha=(\bm\phi,\bm\beta,\gamma)$, the weight depends only on $\thL=(\bm\phi,\bm\beta)$.

For the attaining rule, fix an operational interval $\Phi_{\rm op}=[\phi_-,\phi_+]$, power bounds $0<p_-<p_+$, noise bounds $0<\sigma_-^2<\sigma_+^2$, and a source separation $\Delta>0$. These bounds are fixed before sampling. In the $(\bm\phi,\bm\beta,\gamma)$ chart, define the compact parameter set $\Theta$ by
\begin{equation}\label{eq:Theta-chart}
\begin{gathered}
\phi_-\le\phi_1,\qquad \phi_K\le\phi_+,\\
\phi_{j+1}-\phi_j\ge\Delta\quad(1\le j<K),
\qquad
\phi_K-\phi_1\le1-\Delta,\\
\log\sigma_-^2\le\gamma\le\log\sigma_+^2,
\qquad
\log p_-\le\beta_j+\gamma\le\log p_+.
\end{gathered}
\end{equation}
For $K=1$, the two gap constraints in (\ref{eq:Theta-chart}) are vacuous. We choose the bounds so that $\Theta$ has nonempty interior. It is a compact convex polytope contained in the regular ordered chart, and every interior regular scene can be contained in a set of this form. With $N$ and $K$ fixed, Lemma~\ref{lem:compactreg} gives bounded covariance eigenvalues, uniformly nonsingular Fisher information, and score moments on $\Theta$. No quantitative lower bound on $N\Delta$ is required for this construction.

Since $\partial_\gamma\psi=\bm0$, the efficient Fisher information (\ref{eq:general-Jeff}) for $\thL$ is
\begin{equation}\label{eq:Jeff}
\Jb_L
=\Jb_{(\phi\beta)}
-\Jb_{(\phi\beta)\gamma}\Jb_{\gamma\gamma}^{-1}\Jb_{\gamma(\phi\beta)},
\end{equation}
where $\Jb_{(\phi\beta)}$, $\Jb_{(\phi\beta)\gamma}$, and $\Jb_{\gamma\gamma}$ are, respectively, the $\thL$-by-$\thL$, $\thL$-by-$\gamma$, and $\gamma$-by-$\gamma$ blocks of the full Fisher information. For real coordinates indexed by $a$ and $b$, the per-snapshot Fisher information has the Slepian--Bangs form
$\Jb_{ab}=\tr(\Rb^{-1}\partial_a\Rb\,\Rb^{-1}\partial_b\Rb)$. The loss-curvature matrix $\Mb\in\R^{2K\times2K}$, specializing (\ref{eq:general-WM}), is
\begin{equation}\label{eq:Mdef}
\Mb_{ab}
=\frac{\operatorname{Re}\langle\partial_a\wopt,\partial_b\wopt\rangle_{\Rb}}{P_\star},
\end{equation}
where $a$ and $b$ now index the $2K$ coordinates of $\thL$. Define
\begin{equation*}
C_\theta=\tr(\Mb\Jb_L^{-1})=\sum_{s=1}^{2K}\lambda_s,
\qquad
\lambda_s=\lambda_s\!\left(\Jb_L^{-1/2}\Mb\Jb_L^{-1/2}\right).
\end{equation*}
The eigenvalues $\lambda_s$ measure how strongly the locally estimable parameter directions affect the output-SINR loss.

We attain the converse with a split one-step construction. To this end, let $\omega\in(1/2,1)$ be the split exponent. A pilot block of size $m_n=\min\{\lceil n^\omega\rceil,n-1\}$ gives an initial structured estimate. The remaining $n_2=n-m_n$ snapshots provide one Fisher-scoring correction. With $\bm\vartheta\in\Theta$ as the trial scene, the pilot is the least-squares fit
\begin{equation}\label{eq:pilotdef}
\tilde\tha_n\in\operatorname*{arg\,min}_{\bm\vartheta\in\Theta}
\|\hat\Rb_1-\Rb(\bm\vartheta)\|_F^2,
\qquad
\hat\Rb_1=\frac1{m_n}\sum_{\ell\le m_n}\x_\ell\x_\ell\herm.
\end{equation}
If the minimizer is not unique, $\tilde\tha_n$ denotes the lexicographically least one. The update block performs one Fisher-scoring step, and the resulting scene estimate is inserted into the MVDR formula:
\begin{equation}\label{eq:onestepdef}
\hat\tha_n=\Pi_\Theta\!\left[\tilde\tha_n+\Jb(\tilde\tha_n)^{-1}S_n(\tilde\tha_n)\right],
\qquad
\hat\w_n=\wopt\!\left(\Rb(\hat\tha_n)\right).
\end{equation}
For a trial parameter $t\in\Theta$ and a real coordinate indexed by $a$, the one-snapshot score coordinate is
\begin{equation*}
\dot\ell_a(\x;t)
=\x\herm\Rb(t)^{-1}\partial_a\Rb(t)\Rb(t)^{-1}\x
-\tr\!\left(\Rb(t)^{-1}\partial_a\Rb(t)\right).
\end{equation*}
Let $\dot\ell(\x;t)$ collect all score coordinates. Then $S_n(t)=n_2^{-1}\sum_{\ell>m_n}\dot\ell(\x_\ell;t)$ is the average update-block score, and $\Pi_\Theta$ is the Euclidean projection onto $\Theta$. The split makes the pilot independent of the score correction. The condition $\omega>1/2$ gives $m_n^{-1/2}=o(n^{-1/4})$, while $\omega<1$ gives $n_2/n\to1$. Appendix~\ref{app:onestep} proves that the estimator is measurable, that the pilot localizes with all moments bounded, and that the one-step expansion holds.

\begin{theorem}[Finite-source ULA limit]\label{thm:main}
Fix an interior scene $\tha\in\Theta_{\rm reg}$ and consider the local alternatives $\tha_{n,\bm h}=\tha+\bm h/\sqrt n$. Then
\begin{equation}\label{eq:structured-converse}
\lim_{r\to\infty}\liminf_{n\to\infty}
\inf_{\hat\w_n}\sup_{\|\bm h\|\le r}
 n\,\E_{\tha_{n,\bm h}}[L(\hat\w_n;\tha_{n,\bm h})]\ge C_\theta.
\end{equation}
If $\tha$ is an interior point of a fixed compact set $\Theta$ in (\ref{eq:Theta-chart}), the split one-step beamformer in (\ref{eq:pilotdef}) and (\ref{eq:onestepdef}) satisfies
\begin{equation}\label{eq:structured-ach}
\sup_{\|\bm h\|\le H}
\left|
 n\,\E_{\tha_{n,\bm h}}[L(\hat\w_n;\tha_{n,\bm h})]-C_\theta
\right|\longrightarrow0
\end{equation}
for every fixed radius $H>0$.
\end{theorem}

Theorem~\ref{thm:main} is proved in Appendix~\ref{app:thm}. Let $\mathcal R_{n,r}^{\rm ULA}(\tha)$ denote the risk in (\ref{eq:local-risk}) for the finite-source model. At every interior point of $\Theta$, Theorem~\ref{thm:main} gives
\begin{equation*}
\begin{split}
\lim_{r\to\infty}\liminf_{n\to\infty}
 n\,\mathcal R_{n,r}^{\rm ULA}(\tha)
&=\lim_{r\to\infty}\limsup_{n\to\infty}
 n\,\mathcal R_{n,r}^{\rm ULA}(\tha)\\
&=C_\theta.
\end{split}
\end{equation*}
Indeed, the converse bounds the first limit from below, while the attaining rule bounds the second from above on every fixed local ball. Thus $C_\theta$ is the sharp first-order local minimax coefficient at every regular interior scene covered by the compact construction.

\begin{proposition}[Finite-source dimension bound]\label{prop:visbound}
At every interior scene $\tha\in\Theta_{\rm reg}$, the loss curvature and the efficient Fisher information satisfy
\begin{equation}\label{eq:matrixbound}
\bm0\preceq\Mb\preceq\frac12\Jb_L.
\end{equation}
Consequently,
\begin{equation}\label{eq:trbound}
0\le\lambda_s\le\frac12,
\qquad
0<C_\theta\le K.
\end{equation}
\end{proposition}

The upper bound depends only on the number of interferers, not on the array size. In particular, $2K+1\le N$ gives $C_\theta\le K\le(N-1)/2<N-1$. The proof is given in Appendix~\ref{app:vis}.

\begin{corollary}[Loss distribution]\label{cor:dist}
Under the achievability conditions of Theorem~\ref{thm:main}, let $\bm h_n$ be any sequence of local shifts satisfying $\|\bm h_n\|\le H$, and define
$L_n=L(\hat\w_n;\tha+\bm h_n/\sqrt n)$. Then
\begin{equation}\label{eq:lossdist}
nL_n\ \Rightarrow\ \sum_{s=1}^{2K}\lambda_s Z_s^2.
\end{equation}
Here $Z_1,\ldots,Z_{2K}$ are independent standard normal variables, and all moments converge.
\end{corollary}

The eigenvalues $\lambda_s$ determine the limiting mean, first-order quantiles, and outage probabilities. The proof is given in Appendix~\ref{app:thm}.

\begin{remark}[Risk conventions]\label{rem:risk}
Write $e=E/P_\star$, so that $L=e/(1+e)$. Note that the beamforming literature also uses the expected excess output power $\E[e]$ and the ratio-of-expectations loss $\E[e]/(1+\E[e])$. Jensen's inequality gives
\begin{equation}\label{eq:risk-order}
\E[L]\le\frac{\E[e]}{1+\E[e]}\le\E[e].
\end{equation}
Therefore, the converse for the bounded loss also bounds the other two conventions. For the one-step beamformer, the all-moment result in Appendix~\ref{app:thm} gives $n^2\E[e^2]=O(1)$, so the three risks have the same first-order coefficient $C_\theta$.
\end{remark}

\section{The unrestricted Gaussian limit}\label{sec:full}

We next apply Theorem~\ref{thm:general} to the unrestricted complex-Gaussian covariance model
\begin{equation*}
\x_\ell\sim\mathcal{CN}(\bm0,\Rb),\quad 1\le\ell\le n,
\qquad
\Rb\in\mathbb H_{++}^{N}.
\end{equation*}
For $n\ge N$, define the sample covariance $\hat\Rb$ and the SMI weight $\hat\w_{\rm SMI}$ by
\begin{equation*}
\hat\Rb=\frac1n\sum_{\ell=1}^{n}\x_\ell\x_\ell\herm,
\qquad
\hat\w_{\rm SMI}
=\frac{\hat\Rb^{-1}\ba_0}{\ba_0\herm\hat\Rb^{-1}\ba_0}.
\end{equation*}
Note that the sample covariance is positive definite almost surely. Let $\{\Hb_a\}_{a=1}^{N^2}$ be a real basis of Hermitian matrices, and let $\bm\vartheta\in\R^{N^2}$ be the local coordinates of $\Rb$. For covariance $\Rb(\bm\vartheta)$, define $L_{\rm SMI}(\bm\vartheta)=L(\hat\w_{\rm SMI};\Rb(\bm\vartheta))$ and $\rho_{\rm SMI}(\bm\vartheta)=1-L_{\rm SMI}(\bm\vartheta)$. The Gaussian experiment is quadratic-mean differentiable, with
\begin{equation}\label{eq:full-fisher}
(\Jb_{\mathrm{unstr}})_{ab}=\tr(\Rb^{-1}\Hb_a\Rb^{-1}\Hb_b).
\end{equation}
For a Hermitian direction $\Hb$, the differential of the optimum weight is
\begin{equation}\label{eq:full-derivative}
D\psi_{\Rb}[\Hb]
=-\Rb^{-1}\Hb\wopt
+(\ba_0\herm\Rb^{-1}\ba_0)(\wopt\herm\Hb\wopt)\wopt.
\end{equation}
Substituting (\ref{eq:full-fisher}) and (\ref{eq:full-derivative}) into $\tr(\Mb_{\vartheta}\Jb_{\vartheta}^{-1})$ gives a coefficient that does not depend on $\Rb$, mirroring the covariance independence of the RMB law itself.

\begin{theorem}[Unrestricted Gaussian limit]\label{thm:smi}
For a local direction $\bm h$, write $\bm\vartheta_{n,\bm h}=\bm\vartheta+\bm h/\sqrt n$. For $N\ge2$, the general converse in Theorem~\ref{thm:general} evaluates to
\begin{equation}\label{eq:unstr-converse}
\begin{aligned}
&\lim_{r\to\infty}\liminf_{n\to\infty}
\inf_{\hat\w_n}\sup_{\|\bm h\|\le r}
 n\,\E_{\bm\vartheta_{n,\bm h}}
 \big[L(\hat\w_n;\bm\vartheta_{n,\bm h})\big]\ge N-1.
\end{aligned}
\end{equation}
For $n\ge N$, the SMI output-SINR ratio satisfies
\begin{equation}\label{eq:rmb-beta}
\begin{aligned}
\rho_{\rm SMI}(\bm\vartheta)&\sim\operatorname{Beta}(n-N+2,N-1),\\
\E_{\bm\vartheta}[L_{\rm SMI}(\bm\vartheta)]&=\frac{N-1}{n+1}.
\end{aligned}
\end{equation}
The distribution and mean in (\ref{eq:rmb-beta}) are independent of $\Rb$. Hence, for every fixed $r<\infty$ and all sufficiently large $n$,
\begin{equation}\label{eq:unstr-ach}
\sup_{\|\bm h\|\le r}
 n\,\E_{\bm\vartheta_{n,\bm h}}
 \big[L_{\rm SMI}(\bm\vartheta_{n,\bm h})\big]
=\frac{n(N-1)}{n+1}.
\end{equation}
Consequently,
\begin{equation}\label{eq:unstr-expanding-ach}
\lim_{r\to\infty}\lim_{n\to\infty}
\sup_{\|\bm h\|\le r}
 n\,\E_{\bm\vartheta_{n,\bm h}}
 \big[L_{\rm SMI}(\bm\vartheta_{n,\bm h})\big]
=N-1.
\end{equation}
Let $\mathcal R_{n,r}^{\rm full}(\Rb)$ denote the risk in (\ref{eq:local-risk}) in the chosen covariance chart. Then
\begin{equation}\label{eq:unstr-law}
\begin{aligned}
\lim_{r\to\infty}\liminf_{n\to\infty}
 n\,\mathcal R_{n,r}^{\rm full}(\Rb)&=N-1,\\
\lim_{r\to\infty}\limsup_{n\to\infty}
 n\,\mathcal R_{n,r}^{\rm full}(\Rb)&=N-1.
\end{aligned}
\end{equation}
\end{theorem}

The converse in Theorem~\ref{thm:smi} follows by evaluating the general coefficient on the full Hermitian covariance tangent space; the exact distribution in (\ref{eq:rmb-beta}) is the classical RMB law \cite{ReedMallettBrennan1974}. The full proof is given in Appendix~\ref{app:smi}. Since its risk is independent of $\Rb$, SMI is first-order minimax efficient. The theorem complements the RMB law by showing that its leading coefficient is also the best guarantee over all measurable rules in the unrestricted model.

Theorem~\ref{thm:smi} is a local minimax result. A loading or shrinkage rule tuned to one covariance may reduce the risk at that covariance, but it may then pay with a larger risk along nearby alternatives; the local minimax criterion requires the same rule to control all $1/\sqrt n$ alternatives at once.

Theorems~\ref{thm:main} and~\ref{thm:smi} give
\begin{equation}\label{eq:model-gap}
(N-1)-C_\theta>0.
\end{equation}
Both models have the same $1/n$ rate, but their sharp coefficients differ. The gap in (\ref{eq:model-gap}) therefore measures the first-order value of the finite-source model: the reduction from $N-1$ to $C_\theta$ comes from additional structural information, not from any inefficiency of SMI in the model it fits.

\section{Geometry of the finite-source limit}\label{sec:visibility}

Proposition~\ref{prop:visbound} gives the scale of $C_\theta$; we now examine how the coefficient depends on the scene. The whitened representation below separates the information in the training data from the sensitivity of the MVDR weight, and it leads to a source-wise decomposition and a high-INR limit.

In the $\Rb$-whitened basis (coordinates in which the interference-plus-noise covariance becomes the identity), define
\begin{equation*}
\begin{aligned}
\mathbf u&=\frac{\Rb^{-1/2}\ba_0}{\sqrt\zeta},
&\zeta&=\ba_0\herm\Rb^{-1}\ba_0=P_\star^{-1},\\
\mathbf E_a&=\Rb^{-1/2}(\partial_a\Rb)\Rb^{-1/2}.
\end{aligned}
\end{equation*}
Let $\mathbf P_{\mathbf u}^{\perp}=\Ib-\mathbf u\mathbf u\herm$ and $\mathbf E_a^\circ=\mathbf E_a-N^{-1}\tr(\mathbf E_a)\Ib$. Then
\begin{equation}\label{eq:whitenedforms}
\Mb_{ab}=\operatorname{Re}\big[(\mathbf P_{\mathbf u}^{\perp}\mathbf E_a\mathbf u)\herm(\mathbf P_{\mathbf u}^{\perp}\mathbf E_b\mathbf u)\big],
\qquad
(\Jb_L)_{ab}=\tr(\mathbf E_a^\circ\mathbf E_b^\circ).
\end{equation}
The two roles are visible in (\ref{eq:whitenedforms}): $\Jb_L$ measures the full whitened covariance derivative once the global scale is removed, whereas $\Mb$ keeps only the component that changes the MVDR weight in the look direction.

For source $j$, write $\Sb=\Rb^{-1}$ and define
\begin{equation}\label{eq:cosines}
\kappa_j=\frac{|\q_j\herm\Sb\ba_0|^2}
{(\q_j\herm\Sb\q_j)(\ba_0\herm\Sb\ba_0)},
\qquad
\chi_j=\frac{|\ba_j\herm\Sb\ba_0|^2}
{(\ba_j\herm\Sb\ba_j)(\ba_0\herm\Sb\ba_0)}.
\end{equation}
Both quantities are squared cosines in the whitened metric: $\kappa_j$ compares the look direction with the source tangent, and $\chi_j$ compares it with the source steering vector.

The following stronger conditions are used for the dimension-uniform estimates in this section. They are not needed for Theorem~\ref{thm:main} or Proposition~\ref{prop:visbound}.

\begin{assumption}[Quantitative separation]\label{ass:all}
(A1) The powers and noise variance satisfy the bounds $p_j\in[p_-,p_+]$ and $\sigma^2\in[\sigma_-^2,\sigma_+^2]$ defining $\Theta$, with these ranges fixed independently of $N$.

(A2) The source frequencies are separated by at least $\Delta>0$ on the unit torus. The dimensionless separation $\alpha_0=N\Delta$ satisfies $\alpha_0\ge C_0K$, where $C_0$ depends only on a fixed margin $\eta\in(0,1)$ and the ranges in (A1). The look direction is otherwise unrestricted.

(A3) The dimensions satisfy $2K+1\le(1-\eta)N$.
\end{assumption}

Separated frequencies keep the steering vectors nearly orthogonal: the Dirichlet kernel bounds their pairwise inner products by $O(1/\alpha_0)$, the standard route from (A2) to a well-conditioned steering family \cite{Moitra2015,LiLiao2019,LiLiaoFannjiang2020}. Appendix~\ref{app:a1} extends this argument to the tangent vectors $\q_j$ and fixes a sufficient value of $C_0$. For $K\ge2$, packing $K$ frequencies with pairwise gaps $\Delta\ge C_0K/N$ into the unit torus forces $N\gtrsim K^2$, so the quantitatively separated class is nonempty only for such $N$; as noted above, this concerns only the dimension-uniform estimates.

\begin{corollary}[Per-source accounting]\label{cor:persource}
Define
\begin{equation}\label{eq:Dsum}
C_{\rm src}=\sum_{j=1}^K\left[
\frac{\Mb_{\phi_j\phi_j}}{(\Jb_L)_{\phi_j\phi_j}}
+\frac{\Mb_{\beta_j\beta_j}}{(\Jb_L)_{\beta_j\beta_j}}
\right].
\end{equation}
Let
\begin{equation}\label{eq:epsilon-scene}
\begin{aligned}
\mathbf D_L&=\operatorname{diag}\big((\Jb_L)_{11},\ldots,(\Jb_L)_{2K,2K}\big),\\
\varepsilon_\theta&=\left\|\mathbf D_L^{-1/2}(\Jb_L-\mathbf D_L)\mathbf D_L^{-1/2}\right\|.
\end{aligned}
\end{equation}
If $\varepsilon_\theta<1$, then
\begin{equation}\label{eq:coupling}
\left|C_\theta-C_{\rm src}\right|\le\frac{\varepsilon_\theta}{1-\varepsilon_\theta}C_{\rm src}.
\end{equation}
Under Assumption~\ref{ass:all}, Lemma~\ref{lem:diagdom} further gives
\begin{equation}\label{eq:epsilon-upper}
\varepsilon_\theta\le\frac{K-1}{N-1}+C_\varepsilon K\rho_{\rm off},
\qquad
\rho_{\rm off}=O(1/\alpha_0),
\end{equation}
where $C_\varepsilon$ depends only on the power and noise ranges in (A1). This bound is a uniform sufficient estimate. The computable quantity in (\ref{eq:epsilon-scene}) can be used for a given scene. For every source,
\begin{equation}\label{eq:power-ratio}
\frac{\Mb_{\beta_j\beta_j}}{(\Jb_L)_{\beta_j\beta_j}}
=\frac{\chi_j(1-\chi_j)}{1-1/N}.
\end{equation}
\end{corollary}

The quantity $C_{\rm src}$ assigns one angle term and one relative-power term to each source. The coupling bound controls its relative error, which vanishes whenever $\varepsilon_\theta\to0$. In particular, (\ref{eq:epsilon-upper}) gives this conclusion when $K/N\to0$ and $K/\alpha_0\to0$ with the power and noise ranges fixed. Corollary~\ref{cor:persource} follows from Lemma~\ref{lem:diagdom} in Appendix~\ref{app:a1} and the per-source identities in Appendix~\ref{app:vis}.

For one interferer, the angle and relative-power blocks are orthogonal. Let $\kappa=\kappa_1$ and $\chi=\chi_1$ be the two squared cosines in (\ref{eq:cosines}); under the centered ULA convention, the relevant inner products are real, and Appendix~\ref{app:vis} gives
\begin{equation}\label{eq:k1-closed}
C_\theta
=\left[\frac12(\kappa+\chi)-2\kappa\chi\right]
+\frac{\chi(1-\chi)}{1-1/N}.
\end{equation}
The bracketed expression is the angle contribution, and the final term is the relative-power contribution. Both are determined by the whitened cosines in (\ref{eq:cosines}).

\begin{proposition}[High-INR limit]\label{prop:highinr}
Fix the array geometry and the noise variance. Let $p_j=\tau\nu_j$, where the normalized source powers $\nu_j$ are fixed in $[\nu_-,\nu_+]$ with $0<\nu_-\le\nu_+$, and let the common scale $\tau$ tend to infinity. Define the steering matrix $\Ab=[\ba_1,\ldots,\ba_K]$. Assume that $\Ab$ has full column rank, $\mathbf P_{\Ab}^{\perp}\ba_0\ne\bm0$, and $\mathbf P_{\Ab}^{\perp}\q_j\ne\bm0$, where $\mathbf P_{\Ab}^{\perp}$ is the orthogonal projector onto $\operatorname{span}(\Ab)^\perp$. Then
\begin{equation}\label{eq:highinr}
\tr(\Mb\Jb_L^{-1})\longrightarrow\frac12\sum_{j=1}^K\kappa_{j,\infty},
\end{equation}
where
\begin{equation}\label{eq:kappainf}
\kappa_{j,\infty}=
\frac{|(\mathbf P_{\Ab}^{\perp}\q_j)\herm(\mathbf P_{\Ab}^{\perp}\ba_0)|^2}
{\|\mathbf P_{\Ab}^{\perp}\q_j\|^2\|\mathbf P_{\Ab}^{\perp}\ba_0\|^2}.
\end{equation}
\end{proposition}

For a ULA with distinct source directions and $2K+1\le N$, Vandermonde and confluent Vandermonde independence give the conditions in Proposition~\ref{prop:highinr} whenever $\theta_0\notin\{\phi_j\}$. Interestingly, as the INR increases, the source powers and tangent norms enter the loss curvature and the Fisher information at the same order and cancel in the trace. The power-estimation errors remain, but their first-order effect on the output-SINR loss vanishes, so the limit is set by the projected tangent geometry in (\ref{eq:kappainf}) alone.

\begin{corollary}[Look-separation bound]\label{cor:ceiling}
Under Assumption~\ref{ass:all}, suppose in addition that the look direction is separated from every source by at least $\Delta_\ell$ on the unit torus, and let $\alpha_\ell=N\Delta_\ell$. Then
\begin{equation}\label{eq:lookbound}
C_\theta\le C_v\frac{K}{\alpha_\ell^2},
\end{equation}
where $C_v$ depends only on the ranges in (A1) and on $\eta$.
\end{corollary}

Corollary~\ref{cor:ceiling} bounds the sidelobe-scale contribution of interferers separated from the look direction. Proposition~\ref{prop:highinr} and Corollary~\ref{cor:ceiling} are proved in Appendix~\ref{app:vis}.

\section{Numerical results}\label{sec:sim}

This section verifies the sharp coefficients of Theorems~\ref{thm:main} and~\ref{thm:smi}, the geometric description of Section~\ref{sec:visibility}, and the finite-sample onset of the local regime. We simulate a ULA with $N=20$ and $K=3$ interferers at $\bm\phi=(-0.30,0.18,0.35)$. The look direction is broadside, i.e., $\theta_0=0$. The source powers are equal, and the noise variance is one. Unless otherwise specified, the INR is $20$ dB and $n\in\{40,60,100,160,260,420,700\}$. The first value is $n=2N$, the classical short-training point. All risks are computed from $3000$ independent Monte Carlo trials.

SMI is used for the unrestricted model. For the finite-source model, root multiple signal classification (root-MUSIC) is applied to the sample covariance to estimate the $K$ source frequencies \cite{Barabell1983}. The estimates are mapped to $\Phi_{\rm op}$ and sorted. Denote them by $\hat\phi_1,\ldots,\hat\phi_K$. Given these frequencies, root-MUSIC with nonnegative least squares (root-MUSIC$+$LS) estimates the source powers and noise variance from
\begin{equation*}
\min_{\tilde p_j\ge0,\,\tilde\sigma^2\ge0}
\left\|\hat\Rb-\tilde\sigma^2\Ib-
\sum_{j=1}^K\tilde p_j\ba(\hat\phi_j)\ba(\hat\phi_j)\herm\right\|_F^2.
\end{equation*}
The stochastic maximum likelihood (SML) objective is $\log\det\Rb+\tr(\Rb^{-1}\hat\Rb)$ in the ordered $(\bm\phi,\bm\beta,\gamma)$ chart \cite{StoicaNehorai1990}. Eight Fisher-scoring iterations initialized by root-MUSIC$+$LS give the numerical approximation labeled SML in the figures. The log-power chart parametrizes positive powers and noise variance, and the directions are restricted to the ordered operational interval. The coefficient $C_\theta$ is computed from (\ref{eq:R}), (\ref{eq:Jeff}), and (\ref{eq:Mdef}).

\begin{figure*}[!t]
\centering
\includegraphics[width=\textwidth]{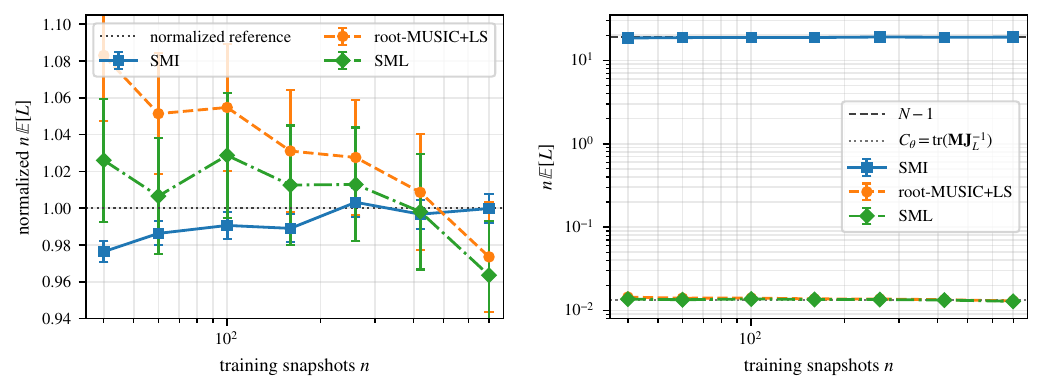}
\caption{First-order coefficients versus the number of training snapshots. Left: the practical structured losses are normalized by the finite-source converse coefficient $C_\theta=\tr(\Mb\Jb_L^{-1})$, while the SMI loss is normalized by its sharp coefficient $N-1$. Unity marks the corresponding first-order reference. Right: the same results on the original logarithmic scale. Error bars show two Monte Carlo standard errors.}
\label{fig:validation}
\end{figure*}

It should be noted that Theorem~\ref{thm:main} establishes achievability with the split one-step rule; here we evaluate the practical root-MUSIC$+$LS and eight-step SML reconstructions against the same benchmark.

Figure~\ref{fig:validation} compares the two model-specific coefficients. The left panel shows the normalized risks, and the right panel shows $n\E[L]$ on the original scale. It is seen that SMI follows the exact value $n(N-1)/(n+1)$ and approaches $N-1=19$. At $n=2N$, the normalized values of root-MUSIC$+$LS and SML are about $1.08$ and $1.03$, respectively, and over the simulated range both structured curves remain between about $0.96$ and $1.08$ relative to $C_\theta\approx0.0133$. The practical methods thus stay close to the finite-source converse coefficient, while SMI agrees with Theorem~\ref{thm:smi}. The error bars do not resolve deviations from $C_\theta$ below a few percent. For this scene, $(N-1)/C_\theta\approx1.4\times10^3$. The coupling measure in (\ref{eq:epsilon-scene}) is $\varepsilon_\theta=0.1053$, with $C_{\rm src}=0.01334408$ and $|C_\theta-C_{\rm src}|/C_\theta=8.49\times10^{-5}$; the bound in (\ref{eq:coupling}) is conservative because it controls the full off-diagonal Fisher block in spectral norm.

\begin{figure}[!t]
\centering
\includegraphics[width=\columnwidth]{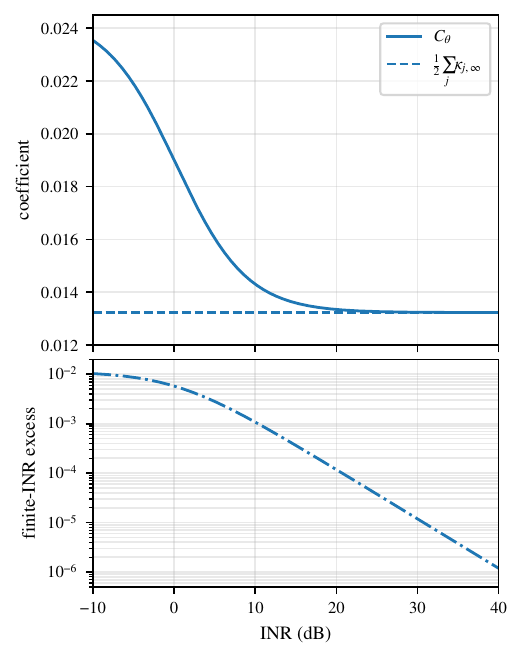}
\caption{Finite-source coefficient versus INR\@. Top: $C_\theta$ and the high-INR limit $\frac12\sum_j\kappa_{j,\infty}$ in Proposition~\ref{prop:highinr}. Bottom: the finite-INR excess $C_\theta-\frac12\sum_j\kappa_{j,\infty}$ on a logarithmic scale.}
\label{fig:visibility}
\end{figure}

Figure~\ref{fig:visibility} examines the high-INR limit. The coefficient $C_\theta$ approaches $\frac12\sum_j\kappa_{j,\infty}\approx0.01323$; above $10$ dB, the finite-INR excess in the lower panel decreases by about one order of magnitude for each additional $10$ dB, and at $20$ dB the difference is $1.18\times10^{-4}$. This trend agrees with Proposition~\ref{prop:highinr}.

\begin{figure}[!t]
\centering
\includegraphics[width=\columnwidth]{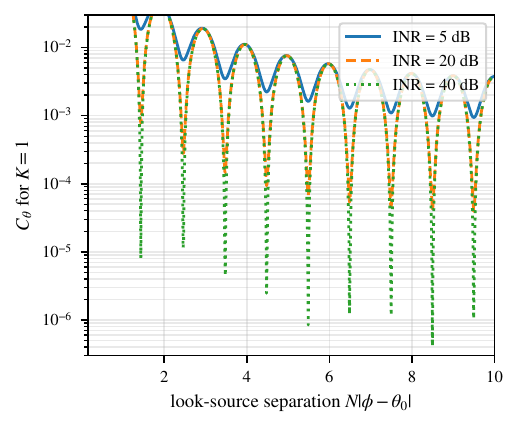}
\caption{The coefficient $C_\theta$ for one interferer versus the look-source separation measured in beamwidths. The oscillations follow the centered-ULA identity in (\ref{eq:k1-closed}). All curves remain far below the finite-source dimension bound $C_\theta\le1$.}
\label{fig:k1}
\end{figure}

Figure~\ref{fig:k1} shows how $C_\theta$ changes as one interferer moves while $N=20$ and the look direction remains fixed. The coefficient varies by several orders of magnitude across the sidelobes and nulls of the array response. Away from the look direction, the curves at $20$ and $40$ dB are nearly the same because the high-INR cancellation is already effective. The sidelobe envelope is consistent with Corollary~\ref{cor:ceiling}, and the oscillations follow (\ref{eq:k1-closed}).

\begin{figure*}[!t]
\centering
\includegraphics[width=0.80\textwidth]{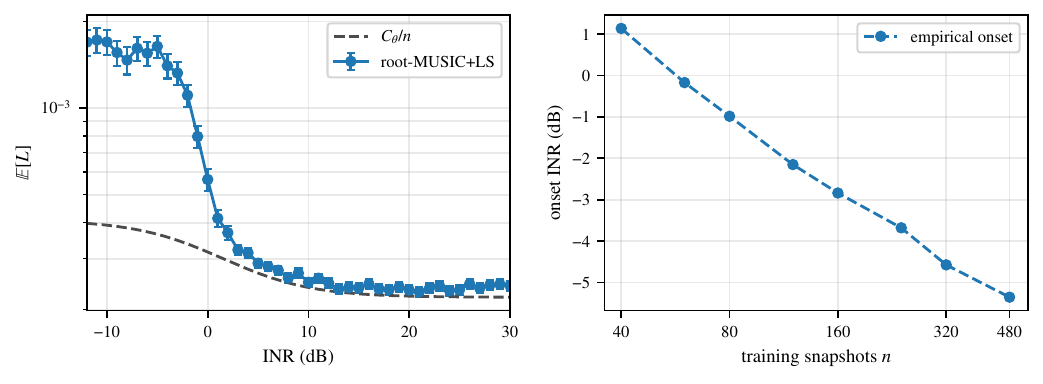}
\caption{Empirical onset of the local regime for root-MUSIC$+$LS. Left: mean loss versus source INR at $n=60$ and the reference $C_\theta/n$. Right: the largest INR for which the loss exceeds $2C_\theta/n$, versus the number of snapshots. The line connects the simulated points.}
\label{fig:threshold}
\end{figure*}

Figure~\ref{fig:threshold} identifies the empirical onset of the local regime for root-MUSIC$+$LS. The array geometry is fixed while the common source INR varies. At $n=60$, the measured loss is close to $C_\theta/n$ at moderate and high INR but increases rapidly once the source powers are too small for reliable localization. We define the empirical onset INR as the largest INR for which the loss exceeds $2C_\theta/n$. It is seen that the onset moves to lower INR as $n$ increases. This experiment characterizes the finite-sample localization stage of root-MUSIC$+$LS, whereas the converse is stated for each fixed regular scene.

\FloatBarrier
\section{Discussion}\label{sec:discussion}

The coefficient $C_\theta$ separates model information from algorithmic efficiency. A practical rule can be compared with $C_\theta/n$ within the finite-source model, while $N-1-C_\theta$ measures the additional first-order cost of fitting an unrestricted covariance. The bound $C_\theta\le K$ shows that this structural gain does not require the interferers to be widely separated. Close sources can nevertheless make the compact-set constants large and delay the finite-sample onset of the local regime. The stronger separation assumptions give explicit control of the geometric approximation and conditioning uniformly in $N$.

The general converse also applies to other calibrated arrays. Let $\mathcal A$ denote the array geometry, let $\xi$ be a direction coordinate on its steering manifold $\ba_{\mathcal A}(\xi)$, and consider $\Rb=\sigma^2\Ib+\sum_jp_j\ba_{\mathcal A}(\xi_j)\ba_{\mathcal A}(\xi_j)\herm$. Theorem~\ref{thm:general} applies whenever the training model is regular and locally identifiable. The one-step argument also gives achievability when the compact-set regularity and pilot-localization conditions in Appendix~\ref{app:onestep} hold. For the ULA, Vandermonde uniqueness supplies identifiability, while Dirichlet-kernel estimates give the additional dimension-uniform bounds.

The finite-source limit treats $K$ as known and fixed. An unknown model order leads to a model-selection problem, and source collisions lead to a singular local model; both cases require different asymptotic analyses.

\section{Conclusion}\label{sec:concl}

This paper established a task-level information limit for adaptive beamforming under finite training. An exact excess-loss identity and the H\'ajek--Le Cam theorem give a local asymptotic minimax converse over all measurable beamforming rules, including biased and irregular rules. A split one-step construction attains the finite-source coefficient $C_\theta$ on fixed compact regular sets, and $0<C_\theta\le K$ holds at every regular ULA scene. In the unrestricted Gaussian model, SMI attains the coefficient $N-1$. These matching results determine the sharp first-order constants and quantify the value of finite-source structure through $N-1-C_\theta$. The source-wise and high-INR formulas explain the coefficient geometrically, while the numerical results illustrate the finite-sample onset of the local regime.

\FloatBarrier
\appendices
\section{Proof of Lemma~\ref{lem:excess}}\label{app:excess}
The proof first removes the irrelevant scale of the candidate weight and then
uses the distortionless constraint to eliminate the cross term. The ratio
\[
\rho(\hat\w)
=\frac{|\hat\w\herm\ba_0|^2}
{(\hat\w\herm\Rb\hat\w)(\ba_0\herm\Rb\inv\ba_0)}
\]
is invariant under $\hat\w\mapsto\alpha\hat\w$ for every $\alpha\ne0$.
Thus, on $\{\hat\w\herm\ba_0\ne0\}$, we may take
$\alpha=1/\overline{\hat\w\herm\ba_0}$ and assume
$\hat\w\herm\ba_0=1$. Put
\[
\bm\delta=\hat\w-\wopt,\qquad
E=\|\bm\delta\|_{\Rb}^{2}.
\]
Both weights are distortionless, hence $\bm\delta\herm\ba_0=0$. Moreover,
the definition of $\wopt$ gives
\[
\Rb\wopt=P_\star\ba_0,
\qquad
\wopt\herm\Rb\wopt=P_\star.
\]
Consequently,
\[
\begin{aligned}
\hat\w\herm\Rb\hat\w
&=\wopt\herm\Rb\wopt
  +2\operatorname{Re}(\bm\delta\herm\Rb\wopt)
  +\bm\delta\herm\Rb\bm\delta\\
&=P_\star
  +2P_\star\operatorname{Re}(\bm\delta\herm\ba_0)
  +E
=P_\star+E.
\end{aligned}
\]
Since $\ba_0\herm\Rb\inv\ba_0=P_\star^{-1}$, the preceding identity gives
\[
\rho(\hat\w)=\frac{P_\star}{P_\star+E},
\qquad
L(\hat\w;\Rb)=1-\rho(\hat\w)
=\frac{E}{P_\star+E}.
\]
Writing $x=E/P_\star$ yields $L=x/(1+x)\le x$ and
$L=x(1+o(1))$ as $x\to0$.

It remains only to cover the part of the sample space excluded by the
normalization. If $\hat\w\herm\ba_0=0$ and $\hat\w\ne\bm0$, the numerator
of \eqref{eq:loss} is zero, so $L=1$. For $\hat\w=\bm0$, the same value is
the convention in \eqref{eq:loss}. Both cases agree with the extension
$E=+\infty$, for which $E/(P_\star+E)=1$ by continuity.
\hfill$\square$

\section{Proof of Theorem~\ref{thm:general}}\label{app:general}

The proof has four steps. We normalize an arbitrary beamformer on the whole sample space, compare the local covariance metric with its value at the center of a fixed neighborhood, apply the local asymptotic minimax theorem to the optimum-weight functional, and transfer the resulting quadratic bound to the exact bounded SINR loss.

Fix the base point $\bm\vartheta$ and write
\begin{equation*}
\psi_0=\psi(\bm\vartheta),
\qquad
P_{\star,0}=P_\star(\bm\vartheta),
\qquad
\mathbf W_0=\frac{\Rb(\bm\vartheta)}{P_{\star,0}}.
\end{equation*}
For a local alternative $\bm\vartheta_{n,\bm h}=\bm\vartheta+\bm h/\sqrt n$, define $\psi_{n,\bm h}$, $P_{\star,n,\bm h}$, and $\mathbf W_{n,\bm h}$ in the same way. The continuity and positive definiteness assumptions imply that, for each fixed $r<\infty$, there is a sequence $\varepsilon_n(r)\to0$ such that, writing $\varepsilon_n=\varepsilon_n(r)$,
\begin{equation}\label{eq:general-metric-compare}
(1-\varepsilon_n)\mathbf W_0
\preceq\mathbf W_{n,\bm h}
\preceq(1+\varepsilon_n)\mathbf W_0,
\qquad
\|\bm h\|\le r.
\end{equation}

For an arbitrary measurable beamformer $\hat\w_n$, define the normalized estimator
\begin{equation}\label{eq:general-normalize}
T_n=
\begin{cases}
\displaystyle\frac{\hat\w_n}{\ba_0\herm\hat\w_n},
&\ba_0\herm\hat\w_n\ne0,\\[7pt]
\displaystyle\frac{\ba_0}{\ba_0\herm\ba_0},
&\ba_0\herm\hat\w_n=0.
\end{cases}
\end{equation}
Then $T_n$ is measurable and finite, and satisfies $T_n\herm\ba_0=1$. On the nonzero-look-gain event, scale invariance gives $L(\hat\w_n)=L(T_n)$. On the zero-look-gain event, the original loss is one and the fallback in (\ref{eq:general-normalize}) has a finite quadratic loss.

At the local truth, put
\begin{equation*}
\ell(T_n)=\|T_n-\psi_{n,\bm h}\|_{\mathbf W_{n,\bm h}}^2,
\qquad
\ell_0(T_n)=\|T_n-\psi_{n,\bm h}\|_{\mathbf W_0}^2.
\end{equation*}
According to (\ref{eq:general-metric-compare}), $\ell(T_n)\ge(1-\varepsilon_n)\ell_0(T_n)$ uniformly on the fixed local ball.

\begin{lemma}[From finite subsets to expanding balls]\label{lem:ball}
Let $\mathcal T$ be a nonempty set, $B\in\R$, and let $F_n:\mathcal T\times\R^d\to[0,\infty)$ be arbitrary. If every sequence $(S_n)\subset\mathcal T$ satisfies
\begin{equation*}
\sup_{I\ \mathrm{finite}}\liminf_n\sup_{\bm h\in I}F_n(S_n,\bm h)\ge B,
\end{equation*}
then
\begin{equation*}
\lim_{r\to\infty}\liminf_n\inf_{T\in\mathcal T}
\sup_{\|\bm h\|\le r}F_n(T,\bm h)\ge B.
\end{equation*}
\end{lemma}

\begin{proof}
Set $V_{n,r}=\inf_{T\in\mathcal T}\sup_{\|\bm h\|\le r}F_n(T,\bm h)$. The quantity $V_{n,r}$ is nondecreasing in $r$. Suppose that the asserted lower bound fails. Then there is an $\epsilon>0$ such that
\begin{equation*}
\lim_{r\to\infty}\liminf_nV_{n,r}<B-3\epsilon.
\end{equation*}
For each integer $m$, choose $n_m>n_{m-1}$ and $S_{n_m}\in\mathcal T$ such that
\begin{equation*}
\sup_{\|\bm h\|\le m}F_{n_m}(S_{n_m},\bm h)<B-\frac{3\epsilon}{2}.
\end{equation*}
Complete $(S_{n_m})$ to a full sequence by assigning a fixed element of $\mathcal T$ at the remaining indices. The hypothesis then gives a finite set $I$ for which
\begin{equation*}
\liminf_n\sup_{\bm h\in I}F_n(S_n,\bm h)\ge B-\epsilon.
\end{equation*}
For all sufficiently large $m$, the set $I$ lies in the ball $\|\bm h\|\le m$, which contradicts the construction of $S_{n_m}$. This completes the proof.
\end{proof}

\begin{proof}[Proof of Lemma~\ref{lem:transfer}]
The action $T_n$ is the normalized estimator in (\ref{eq:general-normalize}), and the loss $\ell_n$ in (\ref{eq:transfer-loss}) is the quantity denoted by $\ell$ below. On the nonzero-look-gain event, Lemma~\ref{lem:excess} gives $nL=n\ell/(1+\ell)$. If $n\ell\le A$, then $1+\ell\le1+A/n$. If $n\ell>A$, monotonicity of $x\mapsto nx/(1+x)$ gives $nL\ge A/(1+A/n)$. On the zero-look-gain event, $nL=n$, which is also no smaller than the right-hand side of (\ref{eq:main-transfer}).
\end{proof}

Combining Lemma~\ref{lem:transfer} and (\ref{eq:general-metric-compare}) gives
\begin{equation}\label{eq:general-transfer}
nL(\hat\w_n)
\ge\frac{1-\varepsilon_n}{1+A/n}
\big[(n\ell_0(T_n))\wedge A\big]
\end{equation}
for all $\|\bm h\|\le r$ and all sufficiently large $n$.

We now apply the local asymptotic minimax theorem. Regard $\Cx^N$ as $\R^{2N}$ and define the bounded bowl-shaped loss
\begin{equation*}
\mathcal L_A(\mathbf z)=\|\mathbf z\|_{\mathbf W_0}^2\wedge A.
\end{equation*}
Quadratic-mean differentiability gives local asymptotic normality with information $\Jb_{\vartheta}$. The optimum-weight functional is differentiable with derivative $\mathbf G_{\vartheta}$. Therefore, the efficient limit variable is
\begin{equation*}
\mathbf Z=\mathbf G_{\vartheta}\mathbf U,
\qquad
\mathbf U\sim\mathcal N(\bm0,\Jb_{\vartheta}^{-1}).
\end{equation*}
For an estimator $S$, define
\begin{equation*}
\begin{split}
\mathcal Q_{n,A}(S,\bm h)
=\E_{\bm\vartheta+\bm h/\sqrt n}\!\Big[&
\big\|\sqrt n\{S-\psi(\bm\vartheta+\bm h/\sqrt n)\}\big\|_{\mathbf W_0}^2\\
&\wedge A\Big].
\end{split}
\end{equation*}
For every estimator sequence $(S_n)$, the finite-subset form of the H\'ajek--Le Cam theorem \cite[Thm.~8.11]{vanderVaart1998} gives
\begin{equation*}
\sup_{I\ \mathrm{finite}}\liminf_n\sup_{\bm h\in I}
\mathcal Q_{n,A}(S_n,\bm h)
\ge\E[\|\mathbf Z\|_{\mathbf W_0}^2\wedge A].
\end{equation*}
Applying Lemma~\ref{lem:ball} and then (\ref{eq:general-transfer}) yields
\begin{equation*}
\begin{split}
&\lim_{r\to\infty}\liminf_n\inf_{\hat\w_n}
\sup_{\|\bm h\|\le r}
\E_{\bm\vartheta+\bm h/\sqrt n}[nL(\hat\w_n)]\\
&\qquad\ge\E[\|\mathbf Z\|_{\mathbf W_0}^2\wedge A].
\end{split}
\end{equation*}
Finally, let $A\to\infty$. Monotone convergence and cyclicity of trace give
\begin{equation*}
\E\|\mathbf Z\|_{\mathbf W_0}^2
=\tr\!\left(\mathbf W_0\mathbf G_{\vartheta}\Jb_{\vartheta}^{-1}\mathbf G_{\vartheta}\herm\right)
=\tr(\Mb_{\vartheta}\Jb_{\vartheta}^{-1}).
\end{equation*}
This proves (\ref{eq:general-converse}). \hfill$\square$

\section{Fisher orthogonality and diagonal dominance of \texorpdfstring{$\Jb_L$}{J L}}\label{app:a1}
This appendix establishes the two facts used later: distinct-node
coherences remain small after whitening, and the efficient Fisher matrix is
a diagonally dominant perturbation of its diagonal.

Throughout, set
\[
\begin{gathered}
\Sb=\Rb\inv,\qquad
\at_j=\Rb^{-1/2}\ba_j,\\
\qt_j=\Rb^{-1/2}\q_j,\qquad
g_j=\ba_j\herm\Sb\ba_j,
\end{gathered}
\]
and let
\[
\rho_{\mathrm{off}}:=\frac{1}{2\alpha_0}.
\]
The covariance derivatives are
\[
\partial_{\phi_j}\Rb
=p_j(\q_j\ba_j\herm+\ba_j\q_j\herm),\qquad
\partial_{p_j}\Rb=\ba_j\ba_j\herm,\qquad
\partial_{\sigma^2}\Rb=\Ib.
\]
We repeatedly use
\[
\tr(\mathbf u\mathbf v\herm\mathbf w\mathbf x\herm)
=(\mathbf v\herm\mathbf w)(\mathbf x\herm\mathbf u)
\]
and the Slepian--Bangs form
\[
\Jb_{ab}=\tr(\Sb\,\partial_a\Rb\,\Sb\,\partial_b\Rb).
\]
All $O(\rho_{\mathrm{off}})$ coherence bounds below are entrywise unless an
operator norm is displayed. Thus a sum of $K$ single-coherence terms costs
$K\rho_{\mathrm{off}}$, whereas a sum of products of two coherences costs
$K\rho_{\mathrm{off}}^2$.

The separation threshold is $\alpha_*=C_0K$. The constant
$C_0=C_0(\eta,p_\pm,\sigma_\pm^2)$ is enlarged once, after the constants in
the estimates below have been fixed, so that
\[
K\rho_{\mathrm{off}}<1,\qquad
\|\mathbf v_\phi\|_\infty^2\le\frac13,
\]
the scene coupling introduced in Lemma~\ref{lem:diagdom} is at most
$1/2$, and the normalized angle diagonal used in the proof of
Corollary~\ref{cor:ceiling} is at least $3/4$. This enlargement depends
only on the ranges in (A1) and on $\eta$.

\begin{lemma}[Dirichlet--Woodbury coherence bounds]\label{lem:coh}
Under Assumption~\ref{ass:all} with $\alpha_0\gtrsim K$, index the half-wavelength ULA symmetrically about $0$, so
$\ba\herm\dot\ba=0$, $\q_j=\dot\ba_j$, and $\|\q_j\|\asymp N$. Let $\Ab=[\ba_1,\dots,\ba_K]$,
$\boldsymbol\Gamma=\Ab\herm\Ab$, $\boldsymbol\Lambda_p=\mathrm{diag}(p_1,\dots,p_K)$, and
$\mathbf B=(\sigma^2\boldsymbol\Lambda_p^{-1}+\boldsymbol\Gamma)\inv$. There is a constant $C$, depending only on
$(p_\pm,\sigma^2_\pm,\eta)$ and uniform over all separations $\alpha_0\ge C_0K$ (in particular free of
$\alpha_0$ itself), such that, for distinct separated nodes $j\ne k$,
\begin{equation*}
|\ba_j\herm\ba_k|\le C\rho_{\mathrm{off}},\quad
\frac{|\q_j\herm\ba_k|}{\|\q_j\|}\le C\rho_{\mathrm{off}},\quad
\frac{|\q_j\herm\q_k|}{\|\q_j\|\|\q_k\|}\le C\rho_{\mathrm{off}},
\end{equation*}
and $\boldsymbol\Gamma=\Ib+O(\rho_{\mathrm{off}})$ entrywise. The operator deviation is a factor of $K$ larger than the entrywise bound,
$\|\boldsymbol\Gamma-\Ib\|=O(K\rho_{\mathrm{off}})$, still below one under $\alpha_*\asymp K$, so
$\boldsymbol\Gamma\succ0$ and $\mathbf B$ has diagonal $\asymp 1$ and off-diagonal
$B_{jk}=O(\rho_{\mathrm{off}})$. The same orders hold after whitening: for $j\ne k$,
$|\at_j\herm\at_k|\le C\rho_{\mathrm{off}}$,
$|\qt_j\herm\at_k|\le C\rho_{\mathrm{off}}\|\qt_j\|\|\at_k\|$, and
$|\qt_j\herm\qt_k|\le C\rho_{\mathrm{off}}\|\qt_j\|\|\qt_k\|$, while $\at_j\herm\at_j\asymp1$
and $\|\qt_j\|^2\asymp\|\q_j\|^2$. The same orders hold with $\Sb^2$ in place of $\Sb$ under either the
$\Sb$- or the $\Sb^2$-weighted normalizer. If, in addition, the
look direction is separated from every source by at least $\Delta_\ell$ on the torus ($\alpha_\ell:=N\Delta_\ell$,
$\rho_\ell:=1/(2\alpha_\ell)$), all bounds above (raw and whitened) also hold with $\ba_0$ in place of
$\ba_k$, at scale $\rho_\ell$ for each pair involving $\theta_0$.
\end{lemma}
\begin{proof}
The proof has three steps: raw Dirichlet estimates, the Woodbury
whitening, and the transfer to $\Sb^2$ and to the look direction.

\emph{Step 1: raw coherences.}
Write $\phi_k-\phi_j=t_{jk}+\ell_{jk}$, where
$t_{jk}\in[-\tfrac12,\tfrac12]$ is a minimal torus representative and
$\ell_{jk}\in\mathbb Z$. Define
\[
D_N(t)=\frac{\sin(\pi Nt)}{N\sin(\pi t)},
\qquad
s_{jk}=(-1)^{(N-1)\ell_{jk}}.
\]
For the symmetrically indexed ULA,
$\ba(\phi+\ell)=(-1)^{(N-1)\ell}\ba(\phi)$, and therefore
\[
\ba_j\herm\ba_k=s_{jk}D_N(t_{jk}).
\]
The sign accounts for the half-integer sensor indices when $N$ is even;
it disappears from all absolute-value bounds below.
Assumption (A2) gives $|t_{jk}|\ge\Delta=\alpha_0/N$. Since
$|\sin(\pi x)|\ge2|x|$ for $|x|\le\tfrac12$,
\begin{equation}\label{eq:dirichlet}
|\ba_j\herm\ba_k|
\le\frac{1}{2N|t_{jk}|}
\le\frac{1}{2N\Delta}
=\frac{1}{2\alpha_0}
=\rho_{\mathrm{off}}.
\end{equation}
Symmetric indexing also gives
$\ba_j\herm\dot\ba_j=N^{-1}\sum_m i2\pi m=0$. Hence
$\q_j=\dot\ba_j$ and
\[
\|\q_j\|^2
=\frac{4\pi^2}{N}\sum_m m^2
\asymp N^2.
\]
The tangent products are derivatives of the same kernel:
\[
\q_j\herm\ba_k=-s_{jk}D_N'(t_{jk}),\qquad
\q_j\herm\q_k=-s_{jk}D_N''(t_{jk}).
\]
Writing $s=\sin(\pi\delta)$, direct differentiation gives
\[
D_N'(\delta)
=\frac{\pi\cos(\pi N\delta)}{s}
-\frac{\pi\cos(\pi\delta)\sin(\pi N\delta)}{Ns^2}
\]
and
\[
\begin{aligned}
D_N''(\delta)
&=-\frac{\pi^2N\sin(\pi N\delta)}{s}
  -\frac{2\pi^2\cos(\pi N\delta)\cos(\pi\delta)}{s^2}\\
&\quad
  +\frac{\pi^2\sin(\pi N\delta)}{Ns}
  +\frac{2\pi^2\sin(\pi N\delta)\cos^2(\pi\delta)}{Ns^3}.
\end{aligned}
\]
On $\alpha_0/N\le|\delta|\le1/2$, the bound
$|s|\ge2|\delta|\ge2\alpha_0/N$, together with $\alpha_0\ge1$, yields
\begin{equation}\label{eq:dirderiv}
|D_N'(\delta)|\le\frac{CN}{\alpha_0},
\qquad
|D_N''(\delta)|\le\frac{CN^2}{\alpha_0}.
\end{equation}
Dividing by $\|\q_j\|\asymp N$ proves the raw normalized coherence
bounds. Equation \eqref{eq:dirichlet} also gives
$\boldsymbol\Gamma=\Ib+O(\rho_{\mathrm{off}})$ entrywise and
\[
\|\boldsymbol\Gamma-\Ib\|
\le (K-1)\rho_{\mathrm{off}}
=O(K\rho_{\mathrm{off}}).
\]

\emph{Step 2: Woodbury whitening.}
Let
\[
\mathbf K_\Gamma=\sigma^2\boldsymbol\Lambda_p^{-1}+\boldsymbol\Gamma,\qquad
\mathbf D_\Gamma=\operatorname{diag}(\mathbf K_\Gamma),
\]
and symmetrically normalize the off-diagonal part:
\[
\mathbf K_\Gamma
=\mathbf D_\Gamma^{1/2}(\Ib+\bm\Xi)\mathbf D_\Gamma^{1/2}.
\]
By (A1), $\mathbf D_\Gamma\asymp\Ib$. The matrix $\bm\Xi$ has zero
diagonal, entrywise order $O(\rho_{\mathrm{off}})$, and
\[
\|\bm\Xi\|\le C K\rho_{\mathrm{off}}<1
\]
after the fixed enlargement of $C_0$. Therefore
\[
\mathbf B
=\mathbf D_\Gamma^{-1/2}
  \sum_{m=0}^{\infty}(-\bm\Xi)^m
  \mathbf D_\Gamma^{-1/2}.
\]
The $m$th off-diagonal contribution is bounded entrywise by
$C^mK^{m-1}\rho_{\mathrm{off}}^m$. The resulting geometric series
shows that
\[
B_{jj}\asymp 1,\qquad
B_{jk}=O(\rho_{\mathrm{off}})\quad(j\ne k).
\]
In particular, $\boldsymbol\Gamma\succ0$ and $\mathbf B$ has the asserted
diagonal/off-diagonal structure.

Woodbury gives
\[
\Sb=\sigma^{-2}(\Ib-\Ab\mathbf B\Ab\herm).
\]
Set
\[
\mathbf c_j=\Ab\herm\ba_j=\mathbf e_j+\bm\epsilon_j,
\qquad
(\bm\epsilon_j)_j=0,\qquad
\|\bm\epsilon_j\|_\infty\le\rho_{\mathrm{off}}.
\]
For $j\ne k$,
\[
\ba_j\herm\Sb\ba_k
=\sigma^{-2}
 \bigl(\ba_j\herm\ba_k-\mathbf c_j\herm\mathbf B\mathbf c_k\bigr),
\]
and the correction has the explicit decomposition
\[
\mathbf c_j\herm\mathbf B\mathbf c_k
=B_{jk}
 +\bm\epsilon_j\herm\mathbf B\mathbf e_k
 +\mathbf e_j\herm\mathbf B\bm\epsilon_k
 +\bm\epsilon_j\herm\mathbf B\bm\epsilon_k.
\]
The first three terms are $O(\rho_{\mathrm{off}})$. The last satisfies
\[
|\bm\epsilon_j\herm\mathbf B\bm\epsilon_k|
\le C\{K\rho_{\mathrm{off}}^2+K^2\rho_{\mathrm{off}}^3\}
=O(\rho_{\mathrm{off}}),
\]
because $K\rho_{\mathrm{off}}=O(1)$. Together with
\eqref{eq:dirichlet}, this proves
$|\at_j\herm\at_k|=|\ba_j\herm\Sb\ba_k|
=O(\rho_{\mathrm{off}})$.

For the tangent terms, put $\mathbf d_j=\Ab\herm\q_j$. Its $j$th
entry is zero, and \eqref{eq:dirderiv} gives
\[
|(\mathbf d_j)_l|
\le C\rho_{\mathrm{off}}\|\q_j\|,
\qquad l\ne j.
\]
Substitution of $\mathbf d_j$ for one or both of the vectors
$\mathbf c_j,\mathbf c_k$ in the preceding expansion yields
\[
|\q_j\herm\Sb\ba_k|
\le C\rho_{\mathrm{off}}\|\q_j\|,
\qquad
|\q_j\herm\Sb\q_k|
\le C\rho_{\mathrm{off}}\|\q_j\|\,\|\q_k\|.
\]
Finally, the Gram bound and (A1) give the Loewner sandwich
\[
\begin{gathered}
\sigma_-^2\Ib
\preceq\Rb
\preceq
\{\sigma_+^2+p_+\|\boldsymbol\Gamma\|\}\Ib,\\
\|\boldsymbol\Gamma\|
\le1+(K-1)\rho_{\mathrm{off}}=O(1).
\end{gathered}
\]
Hence
\[
\|\at_j\|^2=\ba_j\herm\Sb\ba_j\asymp1,
\qquad
\|\qt_j\|^2=\q_j\herm\Sb\q_j\asymp\|\q_j\|^2,
\]
which converts the preceding raw estimates into the normalized
whitened bounds.

\emph{Step 3: $\Sb^2$ and look-direction transfer.}
Squaring the Woodbury representation gives the same finite-rank
structure:
\[
\begin{aligned}
\Sb^2
&=\sigma^{-4}
  \{\Ib-\Ab\mathbf C_2\Ab\herm\},\\
\mathbf C_2
&:=2\mathbf B-\mathbf B\boldsymbol\Gamma\mathbf B
 =\mathbf B+\sigma^2\mathbf B\boldsymbol\Lambda_p^{-1}\mathbf B.
\end{aligned}
\]
The second identity follows from
$\mathbf B\boldsymbol\Gamma\mathbf B
=\mathbf B-\sigma^2\mathbf B\boldsymbol\Lambda_p^{-1}\mathbf B$.
For $j\ne k$, the $l=j,k$ summands in
$(\mathbf B\boldsymbol\Lambda_p^{-1}\mathbf B)_{jk}$ are
$O(\rho_{\mathrm{off}})$, while the remaining sum is
$O(K\rho_{\mathrm{off}}^2)$. On the diagonal, the $l=j$ summand is
$\asymp 1$ and the rest is $O(K\rho_{\mathrm{off}}^2)$. Thus
$\mathbf C_2$ has diagonal $\asymp 1$ and off-diagonal
$O(\rho_{\mathrm{off}})$. Repeating Step~2 with $\mathbf C_2$ proves
all $\Sb^2$ coherence bounds. The Loewner sandwich also gives, uniformly
in $\mathbf x$,
\[
\mathbf x\herm\Sb^2\mathbf x
\asymp
\mathbf x\herm\Sb\mathbf x,
\]
so either family of normalizers may be used.

If the look direction is separated from source $j$, the same Dirichlet
calculation gives the raw source--look estimates at scale $\rho_\ell$.
In the Woodbury expansions of
$\ba_j\herm\Sb\ba_0$, $\q_j\herm\Sb\ba_0$, and their $\Sb^2$
counterparts, every term contains one such source--look factor. All
remaining sums are bounded by $1+O(K\rho_{\mathrm{off}})$. This proves
the whitened bounds at scale $\rho_\ell$. When $\alpha_\ell\le1$, the
same estimates follow directly from Cauchy--Schwarz because
$\rho_\ell\ge1/2$.
\end{proof}

\begin{lemma}\label{lem:a1}
Under (A1)--(A2) with $\alpha_0\gtrsim K$ and the coherence bounds of Lemma~\ref{lem:coh}:
for a single source, $\Jb_{\phi_jp_j}=\Jb_{\phi_j\sigma^2}=0$; for $K\ge2$, both entries are
$O(\rho_{\mathrm{off}}\sqrt{\Jb_{\phi_j\phi_j}\Jb_{bb}})$ ($\Jb_{bb}$ the partner diagonal entry,
$b=p_j$ or $\sigma^2$). Equivalently, the same-source whitened coherence obeys
$|\qt_j\herm\at_j|\le C\rho_{\mathrm{off}}\|\qt_j\|\|\at_j\|$, exactly $0$ at $K=1$. Here
$\rho_{\mathrm{off}}=O(1/\alpha_0)$. Consequently, $\Jb_{\phi_j\beta_k}=\Jb_{\phi_j\gamma}=O(\rho_{\mathrm{off}})$
(after normalization) and $\Jb_L$ is block-diagonal $(\bm\phi)\oplus(\bm\beta)$ up to $O(\rho_{\mathrm{off}})$.
\end{lemma}
\begin{proof}
The same-source trace identities isolate the only quantities that must
be controlled. Using the rank-one trace formula,
\[
\begin{aligned}
\Jb_{\phi_jp_j}
&=p_j\tr\!\left[
\Sb(\q_j\ba_j\herm+\ba_j\q_j\herm)
\Sb\ba_j\ba_j\herm\right]\\
&=2p_jg_j\operatorname{Re}(\ba_j\herm\Sb\q_j),
\end{aligned}
\]
and
\[
\Jb_{\phi_j\sigma^2}
=p_j\tr\!\left[
(\q_j\ba_j\herm+\ba_j\q_j\herm)\Sb^2\right]
=2p_j\operatorname{Re}(\ba_j\herm\Sb^2\q_j).
\]

For $K=1$, Sherman--Morrison gives
\[
\Sb
=\sigma^{-2}
 \left(\Ib-\frac{p_j}{\sigma^2+p_j}\ba_j\ba_j\herm\right).
\]
Because $\ba_j\herm\q_j=0$, one has
$\Sb\q_j=\sigma^{-2}\q_j$ and
$\Sb^2\q_j=\sigma^{-4}\q_j$. Both trace identities therefore vanish
exactly.

For $K\ge2$, Woodbury gives
\[
\ba_j\herm\Sb\q_j
=-\sigma^{-2}
(\Ab\herm\ba_j)\herm\mathbf B(\Ab\herm\q_j).
\]
Here $\Ab\herm\ba_j=\mathbf e_j+O(\rho_{\mathrm{off}})$ entrywise,
whereas $\Ab\herm\q_j$ has zero $j$th coordinate and all other
coordinates of order
$O(\rho_{\mathrm{off}}\|\q_j\|)$. Since $\mathbf B$ has diagonal
$\asymp 1$ and off-diagonal $O(\rho_{\mathrm{off}})$,
\[
|\ba_j\herm\Sb\q_j|
\le
C\rho_{\mathrm{off}}
(\ba_j\herm\Sb\ba_j)^{1/2}
(\q_j\herm\Sb\q_j)^{1/2}.
\]
Equivalently,
\[
|\qt_j\herm\at_j|
\le C\rho_{\mathrm{off}}\|\qt_j\|\,\|\at_j\|.
\]
The $\Sb^2$ representation established in Lemma~\ref{lem:coh} gives,
by the same calculation,
\[
|\ba_j\herm\Sb^2\q_j|
\le
C\rho_{\mathrm{off}}
(\ba_j\herm\Sb^2\ba_j)^{1/2}
(\q_j\herm\Sb^2\q_j)^{1/2}.
\]
The $\Sb$- and $\Sb^2$-normalizers are uniformly equivalent.
Moreover,
\[
\begin{aligned}
\Jb_{\phi_j\phi_j}
&=2p_j^2\left\{\|\qt_j\|^2\|\at_j\|^2
+\operatorname{Re}\big[(\qt_j\herm\at_j)^2\big]\right\}\\
&\asymp p_j^2\|\qt_j\|^2\|\at_j\|^2,\\
\Jb_{p_jp_j}&=g_j^2,\qquad
\Jb_{\sigma^2\sigma^2}=\tr(\Sb^2)\asymp N.
\end{aligned}
\]
The first trace identity is therefore bounded by
\[
C\rho_{\mathrm{off}}
\sqrt{\Jb_{\phi_j\phi_j}\Jb_{p_jp_j}}.
\]
After the uniform $\Sb$--$\Sb^2$ comparison, the second is bounded by
$C\rho_{\mathrm{off}}\sqrt{\Jb_{\phi_j\phi_j}}$, and hence by
\[
C\rho_{\mathrm{off}}
\sqrt{\Jb_{\phi_j\phi_j}\Jb_{\sigma^2\sigma^2}},
\]
because $\Jb_{\sigma^2\sigma^2}\asymp N\ge1$. This proves the two
same-source Fisher bounds in the statement.

For $k\ne j$, another application of the rank-one trace formula gives
\[
\Jb_{\phi_jp_k}
=p_j\left\{
(\ba_j\herm\Sb\ba_k)(\ba_k\herm\Sb\q_j)
+(\q_j\herm\Sb\ba_k)(\ba_k\herm\Sb\ba_j)
\right\}.
\]
Each summand contains two distinct-source whitened coherences, and hence
has normalized order $O(\rho_{\mathrm{off}}^2)$. The change of
coordinates
\[
\partial_{\beta_k}=p_k\partial_{p_k},\qquad
\partial_\gamma
=\sigma^2\partial_{\sigma^2}
 +\sum_{k=1}^K p_k\partial_{p_k}
\]
therefore yields
\[
\Jb_{\phi_j\beta_k}=p_k\Jb_{\phi_jp_k},\qquad
\Jb_{\phi_j\gamma}
=\sigma^2\Jb_{\phi_j\sigma^2}
 +\sum_{k=1}^Kp_k\Jb_{\phi_jp_k}.
\]
The same-source contribution is $O(\rho_{\mathrm{off}})$ after
normalization, and the sum of the cross-source contributions is
$O(K\rho_{\mathrm{off}}^2)=O(\rho_{\mathrm{off}})$. Thus every
unprofiled angle--nuisance entry has the claimed order.

It remains to check that profiling the global scale preserves this
order. Let
\[
\Jb_0=\Jb_{(\phi\beta)},\qquad
\mathbf D_0=\operatorname{diag}(\Jb_0),\qquad
\mathbf v
=\frac{\mathbf D_0^{-1/2}\Jb_{(\phi\beta)\gamma}}
       {\sqrt{\Jb_{\gamma\gamma}}}.
\]
Since $\partial_\gamma\Rb=\Rb$,
\[
\Jb_{\gamma\gamma}=N,\qquad
\Jb_{\beta_k\gamma}=p_kg_k,\qquad
\Jb_{\beta_k\beta_k}=p_k^2g_k^2.
\]
Hence $v_{\beta_k}=N^{-1/2}$ exactly, whereas the preceding bounds give
$v_{\phi_j}=O(\rho_{\mathrm{off}})$. The efficient information is
\[
\Jb_L
=\Jb_0
 -\Jb_{(\phi\beta)\gamma}\Jb_{\gamma\gamma}^{-1}
  \Jb_{\gamma(\phi\beta)},
\]
and its diagonal satisfies
\[
(\mathbf D_L)_{aa}=(\mathbf D_0)_{aa}(1-v_a^2).
\]
By the fixed choice of $C_0$ and (A3), $1-v_a^2$ is bounded away from
zero. Therefore, for every angle--power pair,
\[
\begin{aligned}
|(\Jb_L)_{\phi_j\beta_k}|
&\le
|(\Jb_0)_{\phi_j\beta_k}|\\
&\quad+
\sqrt{(\mathbf D_0)_{\phi_j\phi_j}
      (\mathbf D_0)_{\beta_k\beta_k}}\,
|v_{\phi_j}v_{\beta_k}|.
\end{aligned}
\]
After division by the profiled diagonal normalizers, the two terms have
orders $O(\rho_{\mathrm{off}})$ and
$O(\rho_{\mathrm{off}}/\sqrt N)$, respectively. Hence
\[
\frac{|(\Jb_L)_{\phi_j\beta_k}|}
{\sqrt{(\Jb_L)_{\phi_j\phi_j}(\Jb_L)_{\beta_k\beta_k}}}
=O(\rho_{\mathrm{off}}).
\]
This proves the asserted block decoupling of $\Jb_L$.
\end{proof}

\begin{lemma}[Diagonal dominance of $\Jb_L$]\label{lem:diagdom}
Under Assumption~\ref{ass:all} with $\alpha_0\gtrsim K$, let $\Jb_0=\Jb_{(\phi\beta)}$ be the unprofiled
$(\bm\phi,\bm\beta)$ Fisher block, $\Jb_L=\Jb_0-\Jb_{(\phi\beta)\gamma}\Jb_{\gamma\gamma}\inv\Jb_{\gamma(\phi\beta)}$
($\Jb_{\gamma\gamma}=N$) the profiled efficient information, and $\mathbf D_L=\operatorname{diag}((\Jb_L)_{11},\ldots,(\Jb_L)_{2K,2K})$. Then the
scene coupling $\varepsilon_\theta=\|\mathbf D_L^{-1/2}(\Jb_L-\mathbf D_L)\mathbf D_L^{-1/2}\|$ obeys
\[
\varepsilon_\theta\ \le\ \frac{K-1}{N-1}+C_\varepsilon\,K\rho_{\mathrm{off}},\qquad \rho_{\mathrm{off}}=O(1/\alpha_0),
\]
where $C_\varepsilon$ depends only on the regime parameters of (A1); note that $\tfrac{K-1}{N-1}\le\tfrac KN$ since
$K\le N$. Moreover, as shown along the way, $\varepsilon_\theta=O(\rho_{\mathrm{off}}+K\rho_{\mathrm{off}}^2+K/N)$.
Consequently, for $\varepsilon_\theta<1$ and any $\Mb\succeq0$,
\[
\Big|\tr(\Mb\Jb_L\inv)-\sum_a\frac{\Mb_{aa}}{(\Jb_L)_{aa}}\Big|
\le\frac{\varepsilon_\theta}{1-\varepsilon_\theta}\sum_a\frac{\Mb_{aa}}{(\Jb_L)_{aa}},
\]
equivalently $\tr(\Mb\Jb_L\inv)=(1+O(\varepsilon_\theta))\sum_a\Mb_{aa}/(\Jb_L)_{aa}$.
\end{lemma}
\begin{proof}
The proof proceeds through the normalized score Gram, the unprofiled
Gershgorin bound, the exact rank-one profiling update, and the final trace
comparison.

\emph{Step 1: normalized score entries.}
The whitened angle and relative-power scores are
\[
\mathbf E_{\phi_j}
=p_j(\qt_j\at_j\herm+\at_j\qt_j\herm),
\qquad
\mathbf E_{\beta_j}
=p_j\at_j\at_j\herm.
\]
Expanding the four rank-one products gives
\[
\begin{aligned}
\Jb_{\phi_j\phi_j}
&=\|\mathbf E_{\phi_j}\|_F^2\\
&=2p_j^2
\left\{\|\qt_j\|^2\|\at_j\|^2
       +\operatorname{Re}\big[(\qt_j\herm\at_j)^2\big]\right\}\\
&=2p_j^2\|\qt_j\|^2\|\at_j\|^2
\{1+O(\rho_{\mathrm{off}}^2)\},\\
\Jb_{\beta_j\beta_j}
&=\|\mathbf E_{\beta_j}\|_F^2
=p_j^2\|\at_j\|^4.
\end{aligned}
\]
The same-source cross entry is
\[
\Jb_{\phi_j\beta_j}
=2p_j^2\|\at_j\|^2
\operatorname{Re}(\qt_j\herm\at_j),
\]
which has normalized order $O(\rho_{\mathrm{off}})$ by
Lemma~\ref{lem:a1}.

For $j\ne k$, the cross-source expansions needed below are
\[
\begin{aligned}
\Jb_{\phi_j\phi_k}
&=p_jp_k\big[
(\at_j\herm\qt_k)(\at_k\herm\qt_j)
+(\at_j\herm\at_k)(\qt_k\herm\qt_j)\\
&\hspace{29mm}
+(\qt_j\herm\qt_k)(\at_k\herm\at_j)
+(\qt_j\herm\at_k)(\qt_k\herm\at_j)
\big],\\
\Jb_{\phi_j\beta_k}
&=p_jp_k\big[
(\at_j\herm\at_k)(\at_k\herm\qt_j)
+(\qt_j\herm\at_k)(\at_k\herm\at_j)
\big],\\
\Jb_{\beta_j\beta_k}
&=p_jp_k|\at_j\herm\at_k|^2.
\end{aligned}
\]
Every summand contains two distinct-node coherences. Lemma~\ref{lem:coh}
therefore makes every normalized cross-source entry
$O(\rho_{\mathrm{off}}^2)$.

\emph{Step 2: the unprofiled block.}
Let
\[
\mathbf D_0=\operatorname{diag}(\Jb_0),\qquad
\mathbf F_0
=\mathbf D_0^{-1/2}(\Jb_0-\mathbf D_0)\mathbf D_0^{-1/2}.
\]
Each row of $\mathbf F_0$ has one
$O(\rho_{\mathrm{off}})$ same-source angle--power entry and
$2K-2$ cross-source entries of order
$O(\rho_{\mathrm{off}}^2)$. Since $\mathbf F_0$ is Hermitian,
Gershgorin's theorem gives
\[
\|\mathbf F_0\|
\le C\{\rho_{\mathrm{off}}+K\rho_{\mathrm{off}}^2\}.
\]

\emph{Step 3: exact profiling update.}
Use the vector $\mathbf v$ from the proof of Lemma~\ref{lem:a1}. Then
\[
\begin{aligned}
\Jb_L
&=\mathbf D_0^{1/2}
  (\Ib+\mathbf F_0-\mathbf v\mathbf v\herm)
  \mathbf D_0^{1/2},\\
\mathbf D_L
&=\mathbf D_0^{1/2}
  \{\Ib-\operatorname{diag}(\mathbf v\odot\mathbf v)\}
  \mathbf D_0^{1/2}.
\end{aligned}
\]
Set
\[
\mathbf Q_v
=\{\Ib-\operatorname{diag}(\mathbf v\odot\mathbf v)\}^{-1/2}.
\]
The fixed choice of $C_0$ gives
$1-v_a^2\in[2/3,1]$ and
$\|\mathbf Q_v\|^2\le3/2$. Hence
\[
\begin{aligned}
\mathbf F_L
&:=\mathbf D_L^{-1/2}
   (\Jb_L-\mathbf D_L)
   \mathbf D_L^{-1/2}\\
&=\mathbf Q_v
\left\{
\mathbf F_0
 -\bigl(\mathbf v\mathbf v\herm
        -\operatorname{diag}(\mathbf v\odot\mathbf v)\bigr)
\right\}
\mathbf Q_v.
\end{aligned}
\]

The power components are exactly
$\mathbf v_\beta=N^{-1/2}\bm1$. On the power block the profiling term,
after the two $\mathbf Q_v$ factors, is therefore
\[
-\frac{1}{N-1}(\bm1\bm1^\top-\Ib_K),
\]
whose operator norm is exactly $(K-1)/(N-1)$. The remaining profiling
pieces satisfy
\[
\|\mathbf v_\phi\mathbf v_\phi\herm\|
=O(K\rho_{\mathrm{off}}^2),
\qquad
\|\mathbf v_\phi\mathbf v_\beta\herm\|
=O(K\rho_{\mathrm{off}}/\sqrt N).
\]
Combining these estimates with the bound on $\mathbf F_0$ gives
\[
\begin{aligned}
\varepsilon_\theta:=\|\mathbf F_L\|
&\le
\frac{K-1}{N-1}
+\frac32\|\mathbf F_0\|
+O\!\left(
K\rho_{\mathrm{off}}^2
+\frac{K\rho_{\mathrm{off}}}{\sqrt N}
\right)\\
&\le
\frac{K-1}{N-1}
+C_\varepsilon K\rho_{\mathrm{off}}.
\end{aligned}
\]
Here $C_\varepsilon$ depends only on the ranges in (A1). Also,
\[
\frac{K\rho_{\mathrm{off}}}{\sqrt N}
\le\frac12
\left(K\rho_{\mathrm{off}}^2+\frac KN\right),
\qquad
\frac{K-1}{N-1}\le\frac KN,
\]
so the more explicit bound is
\[
\varepsilon_\theta
=O\!\left(
\rho_{\mathrm{off}}
+K\rho_{\mathrm{off}}^2
+\frac KN
\right).
\]
Under (A3), $(K-1)/(N-1)<(1-\eta)/2$. Enlarging $C_0$ once more, if
needed, makes $C_\varepsilon K\rho_{\mathrm{off}}\le\eta/2$, and hence
$\varepsilon_\theta\le1/2$. The same choice makes the angle factor
$1+O(\rho_{\mathrm{off}}^2)$ above at least $3/4$. These are the two
numerical margins used later.
Since $\mathbf D_L\succ0$ and
$\Ib+\mathbf F_L\succeq(1-\varepsilon_\theta)\Ib\succ0$, this also verifies
$\Jb_L\succ0$ before its inverse is used.

\emph{Step 4: trace comparison.}
Write
\[
\Jb_L
=\mathbf D_L^{1/2}(\Ib+\mathbf F_L)\mathbf D_L^{1/2},
\qquad
\mathbf B_M
=\mathbf D_L^{-1/2}\Mb\mathbf D_L^{-1/2}\succeq0.
\]
Then
\[
\tr(\Mb\Jb_L\inv)
=\tr\{\mathbf B_M(\Ib+\mathbf F_L)\inv\}.
\]
For $\varepsilon_\theta<1$,
\[
\|(\Ib+\mathbf F_L)\inv-\Ib\|
\le\frac{\varepsilon_\theta}{1-\varepsilon_\theta}.
\]
Using
$|\tr(\mathbf B_M\mathbf C)|
\le\|\mathbf C\|\tr(\mathbf B_M)$ for
$\mathbf B_M\succeq0$,
\[
\left|
\tr(\Mb\Jb_L\inv)-\tr(\mathbf B_M)
\right|
\le
\frac{\varepsilon_\theta}{1-\varepsilon_\theta}\tr(\mathbf B_M).
\]
Finally,
\[
\tr(\mathbf B_M)
=\sum_a\frac{\Mb_{aa}}{(\mathbf D_L)_{aa}}
=\sum_a\frac{\Mb_{aa}}{(\Jb_L)_{aa}},
\]
which is the claimed comparison.
\end{proof}

\section{Whitened geometry and proofs of the finite-source bounds}\label{app:vis}
This appendix expresses the loss sensitivity and the efficient Fisher
information in a common whitened basis. The resulting identities prove
the finite-source dimension bound, the per-source ratios, the high-INR limit,
and the look-separation bound.

Write
\[
\Sb=\Rb\inv,\qquad
\zeta=\ba_0\herm\Sb\ba_0=P_\star^{-1},\qquad
\mathbf u=\frac{\Rb^{-1/2}\ba_0}{\sqrt\zeta},
\]
so $\|\mathbf u\|=1$, and define
\[
\mathbf P_{\mathbf u}^{\perp}=\Ib-\mathbf u\mathbf u\herm,\qquad
\mathbf E_a=\Rb^{-1/2}(\partial_a\Rb)\Rb^{-1/2}.
\]

\emph{Whitened sensitivity and Fisher.}
Let $\fb=\Sb\ba_0$, so that $\wopt=\fb/\zeta$. Since
\[
\begin{aligned}
\partial_a\Sb&=-\Sb(\partial_a\Rb)\Sb,\\
\partial_a\zeta
&=-\ba_0\herm\Sb(\partial_a\Rb)\Sb\ba_0\\
&=-\zeta\,\mathbf u\herm\mathbf E_a\mathbf u,
\end{aligned}
\]
differentiation gives
\[
\begin{aligned}
\Rb^{1/2}\partial_a\wopt
&=-\frac1\zeta
  \Rb^{-1/2}(\partial_a\Rb)\Sb\ba_0
  -\frac{\partial_a\zeta}{\zeta^2}
  \Rb^{-1/2}\ba_0\\
&=-\frac1{\sqrt\zeta}
  \{\mathbf E_a\mathbf u-(\mathbf u\herm\mathbf E_a\mathbf u)\mathbf u\}\\
&=-\frac1{\sqrt\zeta}\mathbf P_{\mathbf u}^{\perp}\mathbf E_a\mathbf u.
\end{aligned}
\]
The projection is the differential form of the distortionless identity
$\wopt\herm\ba_0=1$. Dividing the $\Rb$-inner product of two derivatives
by $P_\star=1/\zeta$ yields
\begin{equation}\label{eq:Mwhite}
\Mb_{ab}
=\operatorname{Re}\big[
(\mathbf P_{\mathbf u}^{\perp}\mathbf E_a\mathbf u)\herm
(\mathbf P_{\mathbf u}^{\perp}\mathbf E_b\mathbf u)
\big].
\end{equation}

The Fisher block is the Frobenius Gram
$\Jb_{ab}=\tr(\mathbf E_a\mathbf E_b)$. Since
$\partial_\gamma\Rb=\Rb$, the whitened scale score is
$\mathbf E_\gamma=\Ib$, so
\[
\Jb_{a\gamma}=\tr(\mathbf E_a),\qquad
\Jb_{\gamma\gamma}=N.
\]
Profiling this one-dimensional score gives
\begin{equation}\label{eq:Jwhite}
(\Jb_L)_{ab}
=\tr(\mathbf E_a\mathbf E_b)
 -\frac1N\tr(\mathbf E_a)\tr(\mathbf E_b)
=\tr(\mathbf E_a^\circ\mathbf E_b^\circ),
\end{equation}
where
\[
\mathbf E_a^\circ
=\mathbf E_a-\frac{\tr(\mathbf E_a)}{N}\Ib.
\]

\emph{Proof of Proposition~\ref{prop:visbound}.}
For any Hermitian $\Xb$ and unit vector $\mathbf u$, choose a unitary
matrix $\mathbf Q$ whose first column is $\mathbf u$ and write
\[
\mathbf Q\herm\Xb\mathbf Q
=\begin{bmatrix}\alpha&\mathbf b\herm\\
\mathbf b&\mathbf B\end{bmatrix},
\qquad \alpha\in\R,\quad\mathbf B=\mathbf B\herm.
\]
The task projection keeps exactly the off-diagonal column:
\[
\|\mathbf P_{\mathbf u}^{\perp}\Xb\mathbf u\|^2=\|\mathbf b\|^2,
\qquad
\|\Xb\|_F^2=\alpha^2+2\|\mathbf b\|^2+\|\mathbf B\|_F^2.
\]
Consequently,
\begin{equation}\label{eq:hermitian-task-bound}
\|\mathbf P_{\mathbf u}^{\perp}\Xb\mathbf u\|^2
\le\frac12\|\Xb\|_F^2.
\end{equation}
The factor $1/2$ follows from the two conjugate off-diagonal blocks and
does not require $\tr\Xb=0$.

Since $\mathbf P_{\mathbf u}^{\perp}\Ib\mathbf u=0$, replacing
$\mathbf E_a$ by $\mathbf E_a^\circ$ leaves \eqref{eq:Mwhite} unchanged.
For a real vector $\mathbf z$, put
$\mathbf E_{\mathbf z}^\circ=\sum_a z_a\mathbf E_a^\circ$. Equations
\eqref{eq:Mwhite}, \eqref{eq:Jwhite}, and
\eqref{eq:hermitian-task-bound} give
\[
\begin{aligned}
\mathbf z^\top\Mb\mathbf z
&=\|\mathbf P_{\mathbf u}^{\perp}\mathbf E_{\mathbf z}^\circ\mathbf u\|^2\\
&\le\frac12\|\mathbf E_{\mathbf z}^\circ\|_F^2
=\frac12\mathbf z^\top\Jb_L\mathbf z.
\end{aligned}
\]
Thus $\bm0\preceq\Mb\preceq\Jb_L/2$. Congruence by
$\Jb_L^{-1/2}$ yields $0\le\lambda_s\le1/2$, and summing the $2K$
eigenvalues gives $C_\theta\le K$.

\emph{Positivity of the constant.}
At every interior $\tha\in\Theta_{\rm reg}$,
$\Jb_L\succ0$ by Lemma~\ref{lem:reg}(ii), and $\Mb\succeq0$ by
\eqref{eq:Mwhite}. If
$C_\theta=\tr(\Mb\Jb_L\inv)=0$, then
$\Jb_L^{-1/2}\Mb\Jb_L^{-1/2}$ is positive semidefinite with zero trace, hence
$\Mb=\bm0$. Equation \eqref{eq:Mwhite} then gives
\[
\mathbf P_{\mathbf u}^{\perp}\mathbf E_a\mathbf u=\bm0
\]
for every loss-relevant coordinate.

Let $\fb=\Sb\ba_0$. For the relative-power coordinate,
\[
\mathbf P_{\mathbf u}^{\perp}\mathbf E_{\beta_j}\mathbf u
=p_j(\at_j\herm\mathbf u)\mathbf P_{\mathbf u}^{\perp}\at_j.
\]
If $\phi_j=\theta_0$, then $\at_j=\sqrt\zeta\,\mathbf u$, whereas the angle
coordinate gives
\[
\mathbf P_{\mathbf u}^{\perp}\mathbf E_{\phi_j}\mathbf u
=p_j\sqrt\zeta\,\mathbf P_{\mathbf u}^{\perp}\qt_j\ne\bm0,
\]
because $\q_j\ne\bm0$ and $\q_j\herm\ba_j=0$. Hence $\phi_j\ne\theta_0$.
Distinct ULA steering nodes are not proportional, and invertible
whitening preserves this property, so
$\mathbf P_{\mathbf u}^{\perp}\at_j\ne\bm0$. The vanishing power derivative therefore
forces
\[
\ba_j\herm\fb=0.
\]
With this equality, the angle derivative reduces to a nonzero multiple
of $(\q_j\herm\fb)\at_j$, and hence also forces
$\q_j\herm\fb=0$.

Since $\Rb\fb=\ba_0$ and $\ba_k\herm\fb=0$ for every $k$,
\[
\fb=\frac{\ba_0}{\sigma^2}.
\]
Thus $\ba_j\herm\ba_0=0$ and
$\dot\ba_j\herm\ba_0=0$ for every $j$. Under symmetric indexing,
$\q_j=\dot\ba_j$, while
\[
\ba(\phi)\herm\ba_0=D_N(\theta_0-\phi),
\qquad
D_N(t)=\frac{\sin(N\pi t)}{N\sin(\pi t)}.
\]
Every zero of $D_N$ is simple: at such a zero,
\[
D_N'(t)=\frac{\pi\cos(N\pi t)}{\sin(\pi t)}\ne0.
\]
The two orthogonality conditions cannot therefore hold simultaneously.
This contradiction proves $C_\theta>0$.

\emph{Angle sensitivity.}
Fix source $j$ and put
\[
\begin{gathered}
\at_0=\Rb^{-1/2}\ba_0,\qquad
c_{1j}=\at_j\herm\at_0,\\
c_{2j}=\qt_j\herm\at_0,\qquad
z_j=\qt_j\herm\at_j,\qquad
\mathbf r_j=c_{1j}\qt_j+c_{2j}\at_j.
\end{gathered}
\]
Since $\|\at_0\|^2=\zeta$, the whitened derivative specializes to
\[
\Rb^{1/2}\partial_{\phi_j}\wopt
=-\frac{p_j}{\zeta}
\mathbf P_{\at_0}^{\perp}\mathbf r_j.
\]
The projected norm is
\[
\begin{aligned}
\left\|\mathbf P_{\at_0}^{\perp}\mathbf r_j\right\|^2
&=
|c_{1j}|^2\|\qt_j\|^2
+|c_{2j}|^2\|\at_j\|^2\\
&\quad
+2\operatorname{Re}(\overline{c_{1j}}c_{2j}z_j)
-\frac{4\{\operatorname{Re}(c_{1j}c_{2j}^*)\}^2}{\zeta}.
\end{aligned}
\]
The final term is the distortionless-projection subtraction. By
\eqref{eq:cosines},
\[
|c_{2j}|^2
=\kappa_j\|\qt_j\|^2\zeta,\qquad
|c_{1j}|^2
=\chi_j\|\at_j\|^2\zeta.
\]
Consequently,
\[
\Mb_{\phi_j\phi_j}
=\frac{p_j^2}{\zeta}
\left\|
\mathbf P_{\at_0}^{\perp}(c_{1j}\qt_j+c_{2j}\at_j)
\right\|^2,
\]
whereas
\[
\Jb_{\phi_j\phi_j}
=2p_j^2
\left\{
\|\qt_j\|^2\|\at_j\|^2
+\operatorname{Re}(z_j^2)
\right\}.
\]
This exhibits the cancellation of the common power and aperture factors.
When $K=1$, Lemma~\ref{lem:a1} gives
$z_j=0$ and $\Jb_{\phi\gamma}=0$. Moreover, the Sherman--Morrison formula and $\q_1\herm\ba_1=0$ give
\[
c_1=\frac{\ba_1\herm\ba_0}{\sigma^2+p_1},
\qquad
c_2=\frac{\q_1\herm\ba_0}{\sigma^2}.
\]
Under the centered ULA convention,
$\ba_1\herm\ba_0=D_N(\theta_0-\phi_1)$ and
$\q_1\herm\ba_0=-D_N'(\theta_0-\phi_1)$ are real. Hence
$\{\operatorname{Re}(c_1c_2^*)\}^2=|c_1|^2|c_2|^2$ and
\[
\frac{\Mb_{\phi\phi}}{(\Jb_L)_{\phi\phi}}
=\frac12(\kappa+\chi)-2\kappa\chi.
\]
For $K\ge2$, the same projected-norm identity, together with the
profiled diagonal
$(\Jb_L)_{\phi_j\phi_j}
=\Jb_{\phi_j\phi_j}(1-v_{\phi_j}^2)$, gives the corresponding
per-source angle ratio.

\emph{Relative-power sensitivity.}
Let
\[
c_{0j}=\ba_j\herm\Sb\ba_0.
\]
At fixed $\gamma$, $\partial_{\beta_j}\Rb=p_j\ba_j\ba_j\herm$, and
direct differentiation gives
\[
\partial_{\beta_j}\wopt
=-\frac{p_jc_{0j}}{\zeta}
\left(\Sb\ba_j-\wopt c_{0j}^*\right).
\]
Using $\Rb\wopt=\ba_0/\zeta$,
$g_j=\ba_j\herm\Sb\ba_j$, and
$\ba_j\herm\Sb\ba_0=c_{0j}$,
\[
\left\|\Sb\ba_j-c_{0j}^*\wopt\right\|_{\Rb}^2
=g_j-\frac{|c_{0j}|^2}{\zeta}.
\]
Since
$\chi_j=|c_{0j}|^2/(g_j\zeta)$,
\[
\Mb_{\beta_j\beta_j}
=p_j^2g_j^2\chi_j(1-\chi_j).
\]
On the Fisher side,
\[
\Jb_{\beta_j\beta_j}=p_j^2g_j^2,\qquad
\Jb_{\beta_j\gamma}=p_jg_j,\qquad
\Jb_{\gamma\gamma}=N.
\]
Therefore
\[
(\Jb_L)_{\beta_j\beta_j}
=p_j^2g_j^2\left(1-\frac1N\right),
\]
and, for every $K$,
\[
\frac{\Mb_{\beta_j\beta_j}}
     {(\Jb_L)_{\beta_j\beta_j}}
=\frac{\chi_j(1-\chi_j)}{1-1/N}.
\]
Together with the trace comparison in Lemma~\ref{lem:diagdom}, these
identities give the per-source accounting of
Corollary~\ref{cor:persource}.

\emph{Proof of Proposition~\ref{prop:highinr}.}
Let
\[
\begin{gathered}
p_j=\tau\nu_j,\qquad
\nu_j\in[\nu_-,\nu_+],\\
\mathbf D_\nu=\operatorname{diag}(\nu_1,\ldots,\nu_K),
\qquad
\boldsymbol\Gamma=\Ab\herm\Ab.
\end{gathered}
\]
The powers here follow the separate high-INR asymptotic regime stated in the
proposition and are not restricted by the compact power range in (A1).
All remainders below are uniform over
$\nu_j\in[\nu_-,\nu_+]$, with the array geometry, $N$, $K$, and
$\sigma^2$ fixed.

Because $\Ab$ has full column rank, $\boldsymbol\Gamma\succ0$. Woodbury and a
first-order inverse expansion give
\[
\begin{aligned}
\Sb_\tau
&=\left(\sigma^2\Ib+\tau\Ab\mathbf D_\nu\Ab\herm\right)^{-1}\\
&=\sigma^{-2}\mathbf P_{\Ab}^{\perp}
 +\tau^{-1}
  \Ab\boldsymbol\Gamma^{-1}\mathbf D_\nu^{-1}
  \boldsymbol\Gamma^{-1}\Ab\herm
 +O(\tau^{-2}),
\end{aligned}
\]
where
\[
\mathbf P_{\Ab}^{\perp}
=\Ib-\Ab\boldsymbol\Gamma^{-1}\Ab\herm.
\]
Since $\mathbf P_{\Ab}^{\perp}\ba_j=0$ and
$\Ab\herm\ba_j=\boldsymbol\Gamma\mathbf e_j$,
\[
\Sb_\tau\ba_j
=\tau^{-1}\nu_j^{-1}\Ab\boldsymbol\Gamma^{-1}\mathbf e_j
+O(\tau^{-2}),
\]
and hence
\[
\ba_j\herm\Sb_\tau\ba_k
=(\tau\nu_j)^{-1}\delta_{jk}+O(\tau^{-2}),
\qquad
\q_j\herm\Sb_\tau\ba_k=O(\tau^{-1}).
\]
Also,
\[
\Sb_\tau\q_j
=\sigma^{-2}\mathbf P_{\Ab}^{\perp}\q_j+O(\tau^{-1}),
\]
so
\[
\q_j\herm\Sb_\tau\q_k
=\sigma^{-2}
(\mathbf P_{\Ab}^{\perp}\q_j)\herm
(\mathbf P_{\Ab}^{\perp}\q_k)
+O(\tau^{-1}).
\]
The hypothesis $\mathbf P_{\Ab}^{\perp}\ba_0\ne0$ gives
\[
\zeta_\tau
=\ba_0\herm\Sb_\tau\ba_0
\longrightarrow
\zeta_\infty
:=\sigma^{-2}\|\mathbf P_{\Ab}^{\perp}\ba_0\|^2>0.
\]
Thus, with the notation used above,
\[
\begin{aligned}
g_j
&=(\tau\nu_j)^{-1}+O(\tau^{-2}),\\
\|\qt_j\|^2
&=\sigma^{-2}\|\mathbf P_{\Ab}^{\perp}\q_j\|^2+O(\tau^{-1}),\\
z_j&=O(\tau^{-1}),\qquad
c_{1j}=O(\tau^{-1}),\\
c_{2j}
&=\sigma^{-2}
(\mathbf P_{\Ab}^{\perp}\q_j)\herm
(\mathbf P_{\Ab}^{\perp}\ba_0)
+O(\tau^{-1}).
\end{aligned}
\]

Substitution into the angle Fisher diagonal and the projected sensitivity
norm yields
\[
\Jb_{\phi_j\phi_j}
=2\tau\nu_j\sigma^{-2}
\|\mathbf P_{\Ab}^{\perp}\q_j\|^2+O(1)
\]
and
\[
\Mb_{\phi_j\phi_j}
=\tau\,
\frac{\nu_j
\left|
\sigma^{-2}(\mathbf P_{\Ab}^{\perp}\q_j)\herm
(\mathbf P_{\Ab}^{\perp}\ba_0)
\right|^2}
{\zeta_\infty}
+O(1).
\]
Define
\[
\begin{aligned}
\mathbf D_J
&=\operatorname{diag}\!\left(
2\nu_j\sigma^{-2}
\|\mathbf P_{\Ab}^{\perp}\q_j\|^2
\right),\\
\mathbf D_M
&=\operatorname{diag}\!\left(
\nu_j\sigma^{-2}
\frac{
\left|
(\mathbf P_{\Ab}^{\perp}\q_j)\herm
(\mathbf P_{\Ab}^{\perp}\ba_0)
\right|^2}
{\|\mathbf P_{\Ab}^{\perp}\ba_0\|^2}
\right).
\end{aligned}
\]
The assumption $\mathbf P_{\Ab}^{\perp}\q_j\ne0$ makes
$\mathbf D_J\succ0$.

For $j\ne k$,
\[
\at_j\herm\at_k=O(\tau^{-2}),\qquad
\qt_j\herm\at_k=O(\tau^{-1}),\qquad
\qt_j\herm\qt_k=O(1).
\]
The four-term Fisher expansion in Lemma~\ref{lem:diagdom} then shows
that every off-diagonal angle entry is $O(1)$ after multiplication by
$p_jp_k\asymp\tau^2$. The same order holds for the off-diagonal
angle sensitivity. For the vector $\mathbf r_j$ defined above, each of
the four terms in $\mathbf r_j\herm\mathbf r_k$ is
$O(\tau^{-2})$, and the rank-one subtraction induced by
$\mathbf P_{\at_0}^{\perp}$ has the same order. Multiplication by
$p_jp_k/\zeta_\tau$ therefore leaves $O(1)$.
Moreover,
$\Jb_{\phi_j\gamma}
=2p_j\operatorname{Re}(\qt_j\herm\at_j)=O(1)$, so profiling the global
scale changes the angle block only by $O(1)$ and does not alter its
$\tau\mathbf D_J$ leading term. The angle--power Fisher block is also
$O(1)$ by the same rank-one trace expansions.

For the relative-power block,
$p_jg_j=1+O(\tau^{-1})$, and
\[
\chi_j
=\frac{|c_{0j}|^2}{g_j\zeta_\tau}
=O(\tau^{-1}).
\]
Hence
$\Mb_{\beta\beta}=O(\tau^{-1})$, while the positive-semidefinite Cauchy--Schwarz inequality gives
$\Mb_{\phi\beta}=O(1)$. The Fisher power block satisfies
\[
[\Jb_L]_{\beta\beta}
=\Ib_K-\frac1N\bm1\bm1^\top+O(\tau^{-1}).
\]
Indeed,
\[
\begin{aligned}
\Jb_{\beta_j\beta_j}
&=p_j^2g_j^2=1+O(\tau^{-1}),\\
\Jb_{\beta_j\beta_k}
&=p_jp_k|\ba_j\herm\Sb_\tau\ba_k|^2
=O(\tau^{-2}),\qquad j\ne k,\\
\Jb_{\beta_j\gamma}
&=p_jg_j=1+O(\tau^{-1}).
\end{aligned}
\]
These entries give the displayed scale-profiled block.
The condition $\mathbf P_{\Ab}^{\perp}\ba_0\ne0$ forces $K<N$, so the limiting
matrix is positive definite, with minimum eigenvalue $1-K/N$.
Collecting the block orders,
\[
\Jb_L
=
\begin{bmatrix}
\tau\mathbf D_J+O(1)&O(1)\\
O(1)&\Ib_K-\frac1N\bm1\bm1^\top+O(\tau^{-1})
\end{bmatrix}.
\]
Similarly,
\[
\Mb
=
\begin{bmatrix}
\tau\mathbf D_M+O(1)&O(1)\\
O(1)&O(\tau^{-1})
\end{bmatrix}.
\]

Block inversion now gives
\[
\begin{aligned}
{[\Jb_L\inv]}_{\phi\phi}
&=\tau^{-1}\mathbf D_J^{-1}
  \{\Ib+O(\tau^{-1})\},\\
{[\Jb_L\inv]}_{\phi\beta}
&=O(\tau^{-1}),\qquad
{[\Jb_L\inv]}_{\beta\beta}=O(1).
\end{aligned}
\]
Therefore the power and cross contributions to
$\tr(\Mb\Jb_L\inv)$ are $O(\tau^{-1})$, while
\[
\tr\!\left(
\Mb_{\phi\phi}{[\Jb_L\inv]}_{\phi\phi}
\right)
=\tr(\mathbf D_M\mathbf D_J^{-1})+O(\tau^{-1}).
\]
Finally,
\[
\begin{aligned}
\tr(\mathbf D_M\mathbf D_J^{-1})
&=\frac12\sum_{j=1}^K
\frac{
\left|
(\mathbf P_{\Ab}^{\perp}\q_j)\herm
(\mathbf P_{\Ab}^{\perp}\ba_0)
\right|^2}
{\|\mathbf P_{\Ab}^{\perp}\q_j\|^2
 \|\mathbf P_{\Ab}^{\perp}\ba_0\|^2}\\
&=\frac12\sum_{j=1}^K\kappa_{j,\infty},
\end{aligned}
\]
which proves Proposition~\ref{prop:highinr}.
\hfill$\square$

\emph{Proof of Corollary~\ref{cor:ceiling}.}
By the look-separation clause of Lemma~\ref{lem:coh},
\[
|c_{1j}|
\le C\rho_\ell\|\at_j\|\,\|\at_0\|,
\qquad
|c_{2j}|
\le C\rho_\ell\|\qt_j\|\,\|\at_0\|.
\]
Thus
\[
\chi_j\le C^2\rho_\ell^2,\qquad
\kappa_j\le C^2\rho_\ell^2.
\]
Dropping the nonnegative projection subtraction in the exact angle
formula and using $\|\mathbf x+\mathbf y\|^2
\le2\|\mathbf x\|^2+2\|\mathbf y\|^2$ gives
\[
\Mb_{\phi_j\phi_j}
\le
2p_j^2\|\qt_j\|^2\|\at_j\|^2
(\kappa_j+\chi_j).
\]
From the profiled diagonal in Lemma~\ref{lem:diagdom},
$1-v_{\phi_j}^2\ge2/3$, while the fixed choice of $C_0$ gives
$1+O(\rho_{\mathrm{off}}^2)\ge3/4$. Hence
\[
(\Jb_L)_{\phi_j\phi_j}
\ge
p_j^2\|\qt_j\|^2\|\at_j\|^2,
\]
and therefore
\[
\frac{\Mb_{\phi_j\phi_j}}
     {(\Jb_L)_{\phi_j\phi_j}}
\le2(\kappa_j+\chi_j)
\le4C^2\rho_\ell^2.
\]
The exact power ratio gives, since $N\ge2$,
\[
\frac{\Mb_{\beta_j\beta_j}}
     {(\Jb_L)_{\beta_j\beta_j}}
=\frac{\chi_j(1-\chi_j)}{1-1/N}
\le2\chi_j
\le2C^2\rho_\ell^2.
\]
Summing the $2K$ diagonal ratios gives at most
$6C^2K\rho_\ell^2$. Lemma~\ref{lem:diagdom}, with the established bound
$\varepsilon_\theta\le1/2$, converts this diagonal sum to the trace:
\[
C_\theta
=\tr(\Mb\Jb_L\inv)
\le12C^2K\rho_\ell^2
=3C^2\frac{K}{\alpha_\ell^2}.
\]
Taking $C_v=3C^2$ proves the corollary.
\hfill$\square$

\section{Pointwise and uniform regularity}\label{app:reg}
This appendix distinguishes three levels of regularity. Lemma~\ref{lem:reg} holds at every regular scene. Lemma~\ref{lem:compactreg} supplies the fixed-dimensional compact-set bounds needed by the attaining construction. Lemma~\ref{lem:unifreg} adds dimension-uniform conditioning under Assumption~\ref{ass:all}.

\begin{lemma}[Pointwise regularity]\label{lem:reg}
At every interior $\tha\in\Theta_{\rm reg}$: (i) the covariance map is real-analytic and positive definite; (ii) the family $\{\mathcal{CN}(\bm0,\Rb(\tha))\}$ is quadratic-mean differentiable with continuous Fisher information
\[
\Jb_{ab}=\tr(\Rb^{-1}\partial_a\Rb\,\Rb^{-1}\partial_b\Rb),
\]
and $\Jb\succ0$ and $\Jb_L\succ0$; (iii) the $n$-snapshot experiment is locally asymptotically normal at rate $\sqrt n$, with the efficient score for $\thL$ given below; and (iv) the MVDR functional $\psi(\tha)=\wopt(\Rb(\tha))$ is real-analytic and satisfies $\partial_\gamma\psi=0$.
\end{lemma}

\begin{proof}
\emph{(i) Covariance smoothness.}
Every entry of $\ba(\phi)$ is an exponential in $\phi$. Hence $\tha\mapsto\Rb(\tha)$ is real-analytic in the ordered chart. Since $\sigma^2>0$, $\Rb(\tha)\succ0$ at every point of $\Theta_{\rm reg}$, and matrix inversion is real-analytic in a neighborhood of the point.

\emph{(ii) Quadratic-mean differentiability and the information identity.}
Put $\Sb=\Rb\inv$. With respect to Lebesgue measure on $\Cx^N$, one
snapshot has density
\[
    p_\tha(\x)
    =
    \pi^{-N}\det(\Rb)\inv
    \exp(-\x\herm\Sb\x).
\]
Differentiation in a real parameter coordinate $a$ gives the centered score
\[
    \dot\ell_a(\x;\tha)
    =
    \x\herm\Sb(\partial_a\Rb)\Sb\x
    -
    \tr(\Sb\,\partial_a\Rb).
\]
Indeed, $\E_\tha[\x\herm\mathbf C\x]=\tr(\mathbf C\Rb)$, so
$\E_\tha\dot\ell_a=0$. For Hermitian $\mathbf X_1,\mathbf X_2$, the
circular complex-Gaussian fourth-moment identity \cite{Reed1962} reads
\[
    \operatorname{Cov}_\tha
    (\x\herm\mathbf X_1\x,\x\herm\mathbf X_2\x)
    =
    \tr(\Rb\mathbf X_1\Rb\mathbf X_2).
\]
Taking $\mathbf X_a=\Sb(\partial_a\Rb)\Sb$ and using
$\Rb\Sb=\Ib$ on both sides of the product gives
\[
    \operatorname{Cov}_\tha(\dot\ell_a,\dot\ell_b)
    =
    \tr(\Sb\,\partial_a\Rb\,\Sb\,\partial_b\Rb)
    =
    \Jb_{ab}(\tha).
\]
This is the Slepian--Bangs formula in the real parameter chart.

Fix an interior $\tha$. There is an open neighborhood $U$ on which all
powers and the noise variance remain positive and the angle representatives
remain ordered and distinct. On $U$, $\sqrt{p_\vartheta(\x)}$ is continuously differentiable in $\vartheta$ for every $\x$. After shrinking $U$ if necessary, the covariance eigenvalues are bounded above and away from zero. The derivative of $\sqrt{p_\vartheta(\x)}$ is then dominated uniformly on $U$ by $C(1+\|\x\|^2)e^{-c\|\x\|^2}$ for some $c,C>0$. This envelope is square integrable. The quadratic-mean differentiability criterion in \cite[Lem.~7.6]{vanderVaart1998} therefore applies at $\tha$, with score $\dot\ell$ and information $\Jb$.

The following algebraic lemma gives the required pointwise nonsingularity without a quantitative separation condition.

\begin{lemma}[Pointwise nonsingularity without separation]\label{lem:pointwise}
At any interior $\tha$ with distinct (not necessarily separated) angles and $2K+1\le N$, the full Fisher
matrix is nonsingular, $\Jb\succ0$, and $\Jb_L\succ0$ as a Schur complement.
\end{lemma}
\begin{proof}
Whiten the steering and tangent vectors:
\[
    \at_j=\Rb^{-1/2}\ba_j,
    \qquad
    \qt_j=\Rb^{-1/2}\q_j.
\]
By the Slepian--Bangs identity, $\Jb$ is the real Gram matrix, under
$\langle\Xb,\Yb\rangle=\tr(\Xb\Yb)$, of the $2K+1$ Hermitian directions
\[
    p_j(\qt_j\at_j\herm+\at_j\qt_j\herm),
    \qquad
    p_j\at_j\at_j\herm,
    \qquad
    \Ib.
\]
It is thus enough to prove their real linear independence. Suppose
\[
    \mathbf M
    =
    \sum_{j=1}^Kc_jp_j
    (\qt_j\at_j\herm+\at_j\qt_j\herm)
    +
    \sum_{j=1}^Kd_jp_j\at_j\at_j\herm
    +
    e\Ib
    =
    \bm0
\]
for real coefficients $c_j,d_j,e$.

Under the centered array convention, $\q_j=\dot\ba_j$. The $2K$ vectors
$\{\ba_j,\dot\ba_j\}_{j=1}^K$ form a confluent Vandermonde system at $K$
distinct nodes, and hence have rank $2K$ whenever $2K\le N$. Invertible
whitening preserves this rank, so
$\{\at_j,\qt_j\}_{j=1}^K$ are linearly independent. Because
$2K+1\le N$, choose a unit vector $\mathbf z$ orthogonal to their span.
Then
\[
    0=\mathbf z\herm\mathbf M\mathbf z=e,
\]
so the identity component vanishes.

Let $\{\mathbf u_j,\mathbf w_j\}_{j=1}^K$ be the dual basis in
$\operatorname{span}\{\at_j,\qt_j:j\le K\}$, chosen so that
\[
    \mathbf u_j\herm\at_k=\delta_{jk},
    \quad
    \mathbf u_j\herm\qt_k=0,
    \quad
    \mathbf w_j\herm\at_k=0,
    \quad
    \mathbf w_j\herm\qt_k=\delta_{jk}.
\]
Testing the remaining matrix identity on the two pairs
$(\mathbf u_j,\mathbf u_j)$ and $(\mathbf u_j,\mathbf w_j)$ gives
\[
    0=\mathbf u_j\herm\mathbf M\mathbf u_j=p_jd_j,
    \qquad
    0=\mathbf u_j\herm\mathbf M\mathbf w_j=p_jc_j.
\]
Since $p_j>0$, all $c_j$ and $d_j$ vanish. Thus the score directions are
linearly independent, and their Gram matrix satisfies $\Jb\succ0$.
Positive definiteness of the Schur complement $\Jb_L$ follows from the
positive definiteness of the full block matrix.
\end{proof}

\emph{(iii) Local asymptotic normality and the efficient score.}
Quadratic-mean differentiability from part (ii), together with i.i.d.\ sampling, gives local asymptotic normality at rate
$\sqrt n$ \cite[Thm.~7.2]{vanderVaart1998}. For each fixed local shift
$\bm h$,
\[
    \log
    \frac{dP_{\tha+\bm h/\sqrt n}^{\otimes n}}
         {dP_\tha^{\otimes n}}
    =
    \bm h^\top\Delta_{n,\tha}
    -
    \frac12\bm h^\top\Jb(\tha)\bm h
    +
    o_{P_\tha}(1).
\]
Here
\[
    \Delta_{n,\tha}
    =
    \frac1{\sqrt n}\sum_{\ell=1}^n\dot\ell(\x_\ell;\tha),
\qquad
    \Delta_{n,\tha}
    \Rightarrow
    \mathcal N(\bm0,\Jb(\tha)).
\]
Partition $\tha=(\thL,\gamma)$. Orthogonal projection of the
$(\bm\phi,\bm\beta)$ score onto the orthocomplement of the nuisance score
gives
\[
    \dot\ell_L
    =
    \dot\ell_{(\phi\beta)}
    -
    \Jb_{(\phi\beta)\gamma}
    \Jb_{\gamma\gamma}\inv
    \dot\ell_\gamma.
\]
Direct covariance calculation yields
\[
    \operatorname{Cov}_\tha(\dot\ell_L)
    =
    \Jb_{(\phi\beta)}
    -
    \Jb_{(\phi\beta)\gamma}
    \Jb_{\gamma\gamma}\inv
    \Jb_{\gamma(\phi\beta)}
    =
    \Jb_L.
\]
The information is nonsingular by Lemma~\ref{lem:pointwise} and
$\Jb_{\gamma\gamma}=N>0$. Hence the locally asymptotically normal experiment has the nonsingular
limit information required by the local asymptotic minimax theorem used in
Appendix~\ref{app:general}.

\emph{(iv) Analyticity and scale invariance of the MVDR map.}
In the log-power chart,
\[
    \Rb(\tha)
    =
    e^\gamma
    \left(
        \Ib+\sum_{j=1}^Ke^{\beta_j}\ba_j\ba_j\herm
    \right)
    \succ0.
\]
The map $\Rb\mapsto\Rb\inv$ is real-analytic on the open positive-definite
cone, and $\ba_0\herm\Rb\inv\ba_0>0$ there. Therefore
\[
    \psi(\tha)
    =
    \frac{\Rb(\tha)\inv\ba_0}
         {\ba_0\herm\Rb(\tha)\inv\ba_0}
\]
is real-analytic on the whole chart. For every $c>0$,
\[
    \wopt(c\Rb)
    =
    \frac{c^{-1}\Rb\inv\ba_0}
         {c^{-1}\ba_0\herm\Rb\inv\ba_0}
    =
    \wopt(\Rb).
\]
Since $\gamma$ enters $\Rb$ only through the factor $e^\gamma$, the last
identity gives $\partial_\gamma\psi=0$.
This proves Lemma~\ref{lem:reg}.
\end{proof}

\begin{lemma}[Fixed-dimensional compact regularity]\label{lem:compactreg}
Fix $N$ and $K$ with $2K+1\le N$, and let $\Theta$ be the compact set
in (\ref{eq:Theta-chart}). At every $\vartheta\in\Theta$,
\begin{equation}\label{eq:compact-covariance}
\sigma_-^2\Ib\preceq\Rb(\vartheta)
\preceq(\sigma_+^2+Kp_+)\Ib.
\end{equation}
The matrices $\Jb$ and $\Jb_L$ are uniformly positive definite on
$\Theta$, and their inverses are bounded and Lipschitz there. The maps
$\Rb$ and $\psi$ have bounded derivatives of every fixed finite order on
$\Theta$. For every $0<q<\infty$ and $r\in\{0,1,2\}$,
\begin{equation}\label{eq:compact-score-moments}
\sup_{\vartheta\in\Theta}\E_\vartheta
\left[\sup_{t\in\Theta}
\|\partial_t^r\dot\ell(\x;t)\|^q\right]<\infty.
\end{equation}
The constants may depend on $\Theta$, $N$, and $K$; no condition of the
form $N\Delta\ge C_0K$ is imposed.
\end{lemma}
\begin{proof}
Since each steering vector has unit norm,
$\sum_jp_j\ba_j\ba_j\herm\preceq Kp_+\Ib$, which proves
(\ref{eq:compact-covariance}). The linear-independence argument in
Lemma~\ref{lem:pointwise} applies throughout $\Theta$, including its
boundary: every scene has distinct frequencies and positive powers and
noise variance. Hence $\Jb$ and its Schur complement $\Jb_L$ are
positive definite at every point. Continuity and compactness give
\[
\inf_{\vartheta\in\Theta}\lambda_{\min}(\Jb(\vartheta))>0,
\qquad
\inf_{\vartheta\in\Theta}\lambda_{\min}(\Jb_L(\vartheta))>0.
\]
The covariance and MVDR maps are analytic on a neighborhood of $\Theta$.
Their derivatives of each fixed order are therefore bounded. The same
holds for the information matrices, and the inverse identity
$\partial_a\Jb^{-1}=-\Jb^{-1}(\partial_a\Jb)\Jb^{-1}$ gives bounded
derivatives of their inverses. Convexity of $\Theta$ then gives the
Lipschitz bounds.

Each entry of $\partial_t^r\dot\ell(\x;t)$ is a polynomial of degree
at most two in the real and imaginary parts of $\x$, with coefficients
bounded uniformly in $t\in\Theta$. It is thus bounded in norm by
$C_r(1+\|\x\|^2)$. The covariance bound gives Gaussian moments of
all orders uniformly over $\vartheta\in\Theta$, proving
(\ref{eq:compact-score-moments}).
\end{proof}

\begin{lemma}[Dimension-uniform regularity under separation]\label{lem:unifreg}
Under Assumption~\ref{ass:all}, at every interior $\tha\in\Theta$: (i) the covariance satisfies
\[
\sigma_-^2\Ib\preceq\Rb(\tha)\preceq(\sigma_+^2+p_+S_0)\Ib,
\qquad
S_0=1+\frac{K-1}{2\alpha_0};
\]
(ii) with $\mathbf T=\operatorname{diag}(\|\q_1\|,\ldots,\|\q_K\|,1,\ldots,1)$, the rescaled efficient information $\widetilde\Jb_L=\mathbf T^{-1}\Jb_L\mathbf T^{-1}$ satisfies
\[
c_J\Ib\preceq\widetilde\Jb_L\preceq C_J\Ib,
\]
where $c_J$ and $C_J$ depend only on $(\eta,p_\pm,\sigma_\pm^2)$; and (iii) for $\widetilde\Mb=\mathbf T^{-1}\Mb\mathbf T^{-1}$,
\[
\tr(\Mb\Jb_L^{-1})=\tr(\widetilde\Mb\widetilde\Jb_L^{-1})\le K
\]
uniformly in $N$.
\end{lemma}

\begin{proof}
\emph{Uniform covariance bounds.}
Under (A1),
\[
\Rb(\tha)\succeq\sigma^2\Ib\succeq\sigma_-^2\Ib.
\]
For $j\ne k$, the Dirichlet-kernel estimate and (A2) give
\[
|\ba_j\herm\ba_k|
\le \frac{1}{2N|\phi_j-\phi_k|_{\mathbb T}}
\le \frac{1}{2\alpha_0}.
\]
The Gram matrix of $\{\ba_j\}_{j=1}^K$ has unit diagonal. Gershgorin's theorem therefore yields
\[
\left\|\sum_{j=1}^K\ba_j\ba_j\herm\right\|
\le1+\frac{K-1}{2\alpha_0}=S_0.
\]
Since $p_j\le p_+$,
\[
\sigma_-^2\Ib
\preceq\Rb(\tha)
\preceq(\sigma_+^2+p_+S_0)\Ib.
\]
In particular, $\|\Rb(\tha)^{-1}\|\le\sigma_-^{-2}$ uniformly under Assumption~\ref{ass:all}.

\emph{Uniform conditioning in the loss-relevant coordinates.}
The Fisher matrix is the Gram matrix of the whitened covariance directions.
For the loss-relevant coordinates these directions are
\[
    p_j(\qt_j\at_j\herm+\at_j\qt_j\herm),
    \qquad
    p_j\at_j\at_j\herm.
\]
Lemma~\ref{lem:diagdom} gives their normalized Gram after profiling the
global-scale direction. Write
\[
    \mathbf D_L=\operatorname{diag}(\Jb_L),
    \qquad
    \widetilde{\mathbf D}_L
    =
    \mathbf T\inv\mathbf D_L\mathbf T\inv.
\]
Because $\mathbf T$ and $\mathbf D_L$ are diagonal, the normalized Gram is
unchanged by the rescaling:
\[
    \mathbf D_L^{-1/2}\Jb_L\mathbf D_L^{-1/2}
    =
    \widetilde{\mathbf D}_L^{-1/2}
    \widetilde\Jb_L
    \widetilde{\mathbf D}_L^{-1/2}.
\]
Consequently, Lemma~\ref{lem:diagdom} gives
\[
    \widetilde\Jb_L
    =
    \widetilde{\mathbf D}_L^{1/2}
    (\Ib+\mathbf F_L)
    \widetilde{\mathbf D}_L^{1/2},
    \qquad
    \|\mathbf F_L\|
    \le
    \frac{K-1}{N-1}
    +
    C_\varepsilon K\rho_{\mathrm{off}}.
\]
The two terms on the right have different origins and are controlled
separately. From (A3),
\[
    2K+1\le(1-\eta)N
    \quad\Longrightarrow\quad
    \frac{K-1}{N-1}<\frac{1-\eta}{2}.
\]
Moreover, (A2) gives $\alpha_0\ge C_0K$ and
$\rho_{\mathrm{off}}\le C/\alpha_0$. The choice of $C_0$ at the head of
Appendix~\ref{app:a1} ensures
\[
    C_\varepsilon K\rho_{\mathrm{off}}\le\frac{\eta}{2}.
\]
Thus $\|\mathbf F_L\|\le\tfrac12$.

It remains to check the scale of the diagonal. Lemmas~\ref{lem:coh} and~\ref{lem:diagdom} give, for each angle coordinate,
\[
    (\widetilde{\mathbf D}_L)_{\phi_j\phi_j}
    =
    2p_j^2
    \frac{\|\qt_j\|^2}{\|\q_j\|^2}
    \|\at_j\|^2
    \bigl(1+O(\rho_{\mathrm{off}}^2)\bigr)
    (1-v_{\phi_j}^2),
\]
whereas for each relative-power coordinate,
\[
    (\widetilde{\mathbf D}_L)_{\beta_j\beta_j}
    =
    p_j^2g_j^2\left(1-\frac1N\right).
\]
The uniform covariance bounds and Lemma~\ref{lem:coh} imply that each of
$\|\at_j\|^2$, $\|\qt_j\|^2/\|\q_j\|^2$, and $g_j$ lies in
\[
    \left[
        \frac{1}{\sigma_+^2+p_+S_0},
        \frac{1}{\sigma_-^2}
    \right].
\]
Here
$S_0\le1+1/(2C_0)$ under (A2). The profiling factors satisfy
$1-v_{\phi_j}^2\in[2/3,1]$ and
$1-1/N\in[1/2,1]$. The choice of $C_0$ also makes the
$O(\rho_{\mathrm{off}}^2)$ factor lie in a fixed positive interval.
Together with (A1), these bounds yield constants
$0<c_D\le C_D<\infty$, depending only on
$(\eta,p_\pm,\sigma_\pm^2)$, such that
\[
    c_D
    \le
    (\widetilde{\mathbf D}_L)_{aa}
    \le
    C_D
\]
for every loss-relevant coordinate $a$. Combining this diagonal estimate
with $\|\mathbf F_L\|\le\tfrac12$ gives the claimed two-sided bound
\[
    \frac{c_D}{2}\Ib
    \preceq
    \widetilde\Jb_L
    \preceq
    \frac{3C_D}{2}\Ib.
\]
Thus one may take $c_J=c_D/2$ and $C_J=3C_D/2$. At every fixed $N$,
$\mathbf T$ is invertible, so
$\Jb_L=\mathbf T\widetilde\Jb_L\mathbf T\succ0$ as well.

The trace identity is purely algebraic. Since
\[
    \widetilde\Jb_L\inv
    =
    \mathbf T\Jb_L\inv\mathbf T,
    \qquad
    \widetilde\Mb
    =
    \mathbf T\inv\Mb\mathbf T\inv,
\]
cyclicity of trace gives
\[
    \tr(\widetilde\Mb\widetilde\Jb_L\inv)
    =
    \tr\!\left(
        \mathbf T\inv\Mb\Jb_L\inv\mathbf T
    \right)
    =
    \tr(\Mb\Jb_L\inv).
\]
Proposition~\ref{prop:visbound} now gives
$\tr(\widetilde\Mb\widetilde\Jb_L^{-1})=C_\theta\le K$.
This bound holds without the quantitative separation assumption.

This proves Lemma~\ref{lem:unifreg}.
\end{proof}

\section{Split one-step efficiency}\label{app:onestep}
This appendix proves the efficiency of the fixed-chart estimator
\eqref{eq:pilotdef}--\eqref{eq:onestepdef}. Throughout, $N$ and $K$ are fixed,
$\Theta$ is the compact set in (\ref{eq:Theta-chart}), and
$\tha\in\operatorname{int}(\Theta)$. All suprema over $n$ below are taken
for $n$ large enough that the local alternatives lie in $\Theta$.
The proof separates the two
tasks performed by the sample split. The pilot block localizes the
parameter in Euclidean chart distance at rate $m_n^{-1/2}$. Conditional on
that pilot, the update block supplies the score fluctuation, while the
Fisher-scoring correction cancels the pilot error to first order. The
remaining terms are of order
$r_n+\sqrt n\,r_n^2=o(1)$ after $\sqrt n$ scaling.

We may enlarge an admissible localization rate without invalidating the
moment bound. If the bound below holds with $r_n^{(0)}$, define the buffered rate
\[
    r_n=\max\{r_n^{(0)},n^{-1}\}.
\]
Then the same moment bound still holds, and
$r_n=o(n^{-1/4})$ whenever $r_n^{(0)}=o(n^{-1/4})$. We use this
convention throughout the appendix.

\begin{lemma}[Split one-step efficiency]\label{lem:onestep}
Fix $H<\infty$ and write $\tha_n=\tha+\bm h/\sqrt n$, $\thL_n=\thL+\bm h_L/\sqrt n$. Split the $n$
snapshots into a pilot block of size $m_n=o(n)$ and an update block of size $n_2=n-m_n$, and let
$\tilde\tha_n\in\Theta$ be any pilot computed from the pilot block alone that localizes at some rate
$r_n=o(n^{-1/4})$ with all moments: for every $q<\infty$,
\begin{equation}\tag{P}\label{eq:pilotP}
\sup_{\|\bm h\|\le H}\E_{\tha_n}\big\|\tilde\tha_n-\tha_n\big\|^q\ \le\ C_q\,r_n^q.
\end{equation}
For a random sequence $Y_n$, the notation $Y_n=O_{L^q}(a_n)$ means $\|Y_n\|_{L^q}=O(a_n)$, and $o_{L^q}(a_n)$ is defined analogously. Then, uniformly over $\|\bm h\|\le H$, the one-step update \eqref{eq:onestepdef} obeys
\[
\begin{aligned}
&\sqrt n(\hat\thL_n-\thL_n)=\sqrt{\tfrac n{n_2}}\,\Jb_L(\tha_n)\inv\,\mathbb G_{L,n}+\xi_n,\\
&\mathbb G_{L,n}:=n_2^{-1/2}\!\!\sum_{\ell=m_n+1}^{n}\!\dot\ell_L(\x_\ell;\tha_n),
\end{aligned}
\]
where $\xi_n=o_{L^q}(1)$ for every $q<\infty$. Moreover, for each $p<\infty$,
\[
\sup_n\sup_{\|\bm h\|\le H}
\E_{\tha_n}\|\sqrt n(\hat\thL_n-\thL_n)\|^p<\infty.
\]
\end{lemma}

We first record the deterministic and moment bounds used in the lemmas below. Write
\[
    d_n=\tilde\tha_n-\tha_n,
    \qquad
    \Jb_n=\Jb(\tha_n).
\]
Because $\tha$ is an interior point, all $\tha_n$ with
$\|\bm h\|\le H$ lie a positive distance from the boundary of $\Theta$
for sufficiently large $n$. Lemma~\ref{lem:compactreg} gives
\[
\sup_{\vartheta\in\Theta}\|\Jb(\vartheta)^{-1}\|<\infty,
\qquad
\|\Jb(\vartheta)^{-1}-\Jb(\vartheta')^{-1}\|
\le C\|\vartheta-\vartheta'\|,
\]
as well as the score-derivative moment envelopes in
(\ref{eq:compact-score-moments}). It suffices to prove the moment estimates
at orders $q\ge2$, since smaller positive orders follow by monotonicity.
At the true parameter,
\[
    \E_{\tha_n}\dot\ell(\x;\tha_n)=\bm0,
    \qquad
    \E_{\tha_n}
    [\partial_\vartheta\dot\ell(\x;\tha_n)]
    =
    -\Jb_n,
\]
the second identity being Bartlett's identity. Finally, $d_n$ is
measurable with respect to the pilot block, hence independent of the
update-block sigma-field and of every update-block empirical average
evaluated at deterministic parameters.

\begin{lemma}[Pilot localization]\label{lem:pilot}
The lexicographically selected least-squares pilot $\tilde\tha_n$ is measurable. For every
$q<\infty$, uniformly over $\|\bm h\|\le H$, it satisfies
\begin{equation}\label{eq:pilotmom}
\E_{\tha_n}\|d_n\|^q\le C_q\,m_n^{-q/2};
\end{equation}
in particular \eqref{eq:pilotP} holds with $r_n=m_n^{-1/2}$.
\end{lemma}
\begin{proof}
We prove measurability, a global identification gap, and a local inverse
bound.

\emph{Measurability.}
The criterion
\[
    (\mathbf C,\vartheta)
    \longmapsto
    \|\mathbf C-\Rb(\vartheta)\|_F^2
\]
is jointly continuous, and $\Theta$ is compact. The argmin is therefore a
nonempty compact-valued measurable correspondence of $\hat\Rb_1$ by the
measurable maximum theorem \cite[Thm.~18.19]{AliprantisBorder2006}.
Successively minimizing the first coordinate, then the second, and so on
over the surviving compact sets selects the lexicographically least
minimizer. Each step is measurable. Thus $\tilde\tha_n$ is measurable.

\emph{Sample-covariance moments.}
The matrices
$\x_\ell\x_\ell\herm-\Rb(\tha_n)$ are independent and centered, and have
moments of every order, uniformly over the local alternatives. The
fixed-dimensional covariance bound in (\ref{eq:compact-covariance}) and the
Marcinkiewicz--Zygmund inequality give, for every $q\ge2$,
\begin{equation}\label{eq:scmom}
\E_{\tha_n}\|\hat\Rb_1-\Rb(\tha_n)\|_F^q\le C_q\,m_n^{-q/2}.
\end{equation}
The same conclusion for $0<q<2$ follows from Jensen's inequality applied
to the $q=2$ bound.

\emph{Identification.}
The map $\vartheta\mapsto\Rb(\vartheta)$ is injective on the ordered
operational chart. To see this, write
\[
    \Rb(\vartheta)-\sigma^2\Ib
    =
    \sum_{j=1}^Kp_j\ba(\phi_j)\ba(\phi_j)\herm.
\]
The steering vectors are linearly independent because the nodes are
distinct and $K<N$. Hence the signal term is positive semidefinite of
rank $K$, and $\sigma^2$ is the common value of the smallest $N-K$
eigenvalues of $\Rb(\vartheta)$. Once $\sigma^2$ is known, the signal term
is a rank-$K$ positive semidefinite Hermitian Toeplitz matrix. Its
Vandermonde decomposition into $K$ distinct nodes and positive weights is
unique for $K<N$ \cite{YangXieStoica2016}. The ordered
chart fixes the label permutation, so the nodes, powers, and noise variance
are all determined by the covariance.

Fix $\delta>0$ and take $n_0$ large enough that every admissible
$\tha_n$ lies in the interior of the operational chart. Compactness and
the preceding injectivity imply
\begin{equation}\label{eq:idgap}
c_\delta:=\inf_{\substack{\vartheta\in\Theta,\ \|\bm h\|\le H,\ n\ge n_0\\ \|\vartheta-\tha_n\|\ge\delta}}\|\Rb(\vartheta)-\Rb(\tha_n)\|_F>0.
\end{equation}
Indeed, if the infimum were zero, compactness would give subsequences
$\vartheta_s\to\vartheta^\star$ and
$\tha_{n_s}\to\tha^\star$ such that
\[
    \|\vartheta^\star-\tha^\star\|\ge\delta,
    \qquad
    \Rb(\vartheta^\star)=\Rb(\tha^\star).
\]
Injectivity would then force
$\vartheta^\star=\tha^\star$, a contradiction.

\emph{Local inverse bound.}
Let $D(\vartheta)$ denote the Jacobian of
$\vartheta\mapsto\Rb(\vartheta)$, viewed as a linear map from the real
parameter space to Hermitian matrices with the Frobenius norm. The Fisher
matrix is the Gram matrix of the columns of $D(\tha)$ after invertible
whitening. Lemma~\ref{lem:pointwise} therefore implies that $D(\tha)$ is
injective. Set
\[
    s=\frac12\sigma_{\min}(D(\tha))>0.
\]
By continuity of $D$, choose $\delta_0>0$ so that
\[
    \|D(\vartheta)-D(\tha)\|\le s
    \quad\text{whenever}\quad
    \|\vartheta-\tha\|\le2\delta_0,
\]
and enlarge $n_0$ so that
$\|\tha_n-\tha\|\le\delta_0$ uniformly over $\|\bm h\|\le H$. If
$\|\vartheta-\tha_n\|\le\delta_0$, convexity of $\Theta$ keeps the entire
segment between $\tha_n$ and $\vartheta$ in the chart, and
\[
    \Rb(\vartheta)-\Rb(\tha_n)
    =
    \int_0^1
    D\bigl(\tha_n+t(\vartheta-\tha_n)\bigr)
    [\vartheta-\tha_n]\,dt.
\]
Subtracting and adding $D(\tha)$ inside the integral gives
\begin{equation}\label{eq:loclb}
\begin{gathered}
\|\Rb(\vartheta)-\Rb(\tha_n)\|_F
\ge s\,\|\vartheta-\tha_n\|,\\
\|\vartheta-\tha_n\|\le\delta_0.
\end{gathered}
\end{equation}
Indeed, with $v=\vartheta-\tha_n$,
\[
\begin{aligned}
    &\|\Rb(\vartheta)-\Rb(\tha_n)\|_F\\
    &\quad\ge
    \|D(\tha)v\|_F
    -
    \int_0^1
    \|
        [D(\tha_n+tv)-D(\tha)]v
    \|_F\,dt\\
    &\ge
    2s\|v\|-s\|v\|
    =
    s\|v\|.
\end{aligned}
\]
This is the local inverse estimate used below.

Let
\[
    \mathcal E_n
    =
    \left\{
        \|\hat\Rb_1-\Rb(\tha_n)\|_F<c_{\delta_0}/4
    \right\}.
\]
On $\mathcal E_n$, every $\vartheta$ satisfying
$\|\vartheta-\tha_n\|\ge\delta_0$ obeys
\[
    \|\hat\Rb_1-\Rb(\vartheta)\|_F
    \ge
    c_{\delta_0}
    -
    \|\hat\Rb_1-\Rb(\tha_n)\|_F
    >
    \frac34c_{\delta_0},
\]
whereas the residual norm at $\tha_n$ is below $c_{\delta_0}/4$. Hence every
least-squares minimizer lies inside the $\delta_0$-ball. Using
\eqref{eq:loclb}, the triangle inequality, and the minimizing property,
\[
\begin{aligned}
    s\|d_n\|
    &\le
    \|\Rb(\tilde\tha_n)-\Rb(\tha_n)\|_F\\
    &\le
    \|\hat\Rb_1-\Rb(\tilde\tha_n)\|_F
    +
    \|\hat\Rb_1-\Rb(\tha_n)\|_F\\
    &\le
    2\|\hat\Rb_1-\Rb(\tha_n)\|_F.
\end{aligned}
\]
Thus
$\|d_n\|\le(2/s)\|\hat\Rb_1-\Rb(\tha_n)\|_F$ on $\mathcal E_n$.
On its complement, $\|d_n\|\le\operatorname{diam}(\Theta)$. By
\eqref{eq:scmom} and Markov's inequality, for every $p<\infty$,
\[
    \Pr_{\tha_n}(\mathcal E_n^c)\le C_p m_n^{-p/2}.
\]
Combining the two events, choosing $p\ge q$, and applying
\eqref{eq:scmom} proves \eqref{eq:pilotmom}. Finally,
$r_n=m_n^{-1/2}=o(n^{-1/4})$ for
$m_n=\min\{\lceil n^\omega\rceil,n-1\}$ with $\tfrac12<\omega<1$.
\end{proof}

\begin{lemma}[One-step expansion]\label{lem:onestepexp}
Assume $m_n=o(n)$ and that the pilot satisfies \eqref{eq:pilotP} with a
buffered rate $r_n=o(n^{-1/4})$. With
$\bar\tha_n=\tilde\tha_n+\Jb(\tilde\tha_n)\inv S_n(\tilde\tha_n)$ and
$\hat\tha_n=\Pi_\Theta(\bar\tha_n)$, uniformly over $\|\bm h\|\le H$,
\begin{equation}\label{eq:onestep}
\sqrt n(\hat\tha_n-\tha_n)=\Jb_n\inv\sqrt n\,S_n(\tha_n)+o_{L^2}(1).
\end{equation}
Moreover, the $\sqrt n$-scaled remainder is $O_{L^q}(r_n+\sqrt n\,r_n^2)=o_{L^q}(1)$ for every $q<\infty$.
\end{lemma}
\begin{proof}
The proof is the first-order cancellation of the pilot error followed by
moment control of the three residual terms.

Let
\[
    \mathbf K_n
    =
    \frac1{n_2}
    \sum_{\ell>m_n}
    \partial_\vartheta\dot\ell(\x_\ell;\tha_n).
\]
Taylor's formula along the segment from $\tha_n$ to
$\tilde\tha_n$ gives
\[
    S_n(\tilde\tha_n)
    =
    S_n(\tha_n)-\Jb_nd_n+y_n,
    \qquad
    y_n=(\mathbf K_n+\Jb_n)d_n+R_n.
\]
Because the segment lies in the convex set $\Theta$, the second-order
remainder satisfies
\[
    \|R_n\|
    \le
    A_n\|d_n\|^2,
    \qquad
    A_n
    =
    \sup_{\vartheta\in\Theta}
    \left\|
        \frac1{n_2}
        \sum_{\ell>m_n}
        \partial_\vartheta^2\dot\ell(\x_\ell;\vartheta)
    \right\|.
\]
The derivative envelope recorded above implies
$\|A_n\|_{L^q}\le C_q$ for every $q<\infty$.

Substitute the score expansion into the unprojected Fisher-scoring update.
The identity
\[
    d_n-\Jb(\tilde\tha_n)\inv\Jb_nd_n
    =
    \Jb(\tilde\tha_n)\inv
    [\Jb(\tilde\tha_n)-\Jb_n]d_n
\]
exhibits the first-order cancellation. Adding and subtracting
$\Jb_n\inv S_n(\tha_n)$ then gives
\[
    \bar\tha_n-\tha_n
    =
    \Jb_n\inv S_n(\tha_n)
    +
    A_n^{(1)}+A_n^{(2)}+A_n^{(3)},
\]
where
\[
\begin{aligned}
    A_n^{(1)}
    &=
    [\Jb(\tilde\tha_n)\inv-\Jb_n\inv]S_n(\tha_n),\\
    A_n^{(2)}
    &=
    \Jb(\tilde\tha_n)\inv
    [\Jb(\tilde\tha_n)-\Jb_n]d_n,\\
    A_n^{(3)}
    &=
    \Jb(\tilde\tha_n)\inv y_n.
\end{aligned}
\]
These are, respectively, the inverse-information perturbation, the
quadratic pilot term, and the empirical-Hessian/Taylor remainder.

Fix $q<\infty$. Rosenthal's inequality, uniformly over the local
alternatives, gives
\[
    \|S_n(\tha_n)\|_{L^{2q}}
    =
    O(n_2^{-1/2}),
    \qquad
    \|\mathbf K_n+\Jb_n\|_{L^{2q}}
    =
    O(n_2^{-1/2}).
\]
The pilot condition gives
$\|d_n\|_{L^{4q}}\le C_qr_n$. The uniform inverse and Lipschitz bounds
then yield the following estimates. For brevity, let
$B_n=(\mathbf K_n+\Jb_n)d_n$. Then
\[
\begin{aligned}
    \sqrt n\,\|A_n^{(1)}\|_{L^q}
    &\le
    C\sqrt n\,
    \|d_n\|_{L^{2q}}
    \|S_n(\tha_n)\|_{L^{2q}}
    =
    O(r_n),\\
    \sqrt n\,\|A_n^{(2)}\|_{L^q}
    &\le
    C\sqrt n\,\|d_n\|_{L^{2q}}^2
    =
    O(\sqrt n\,r_n^2),\\
    \sqrt n\,\|B_n\|_{L^q}
    &\le
    \sqrt n\,
    \|\mathbf K_n+\Jb_n\|_{L^{2q}}
    \|d_n\|_{L^{2q}}
    =O(r_n),\\
    \sqrt n\,\|R_n\|_{L^q}
    &\le
    \sqrt n\,
    \|A_n\|_{L^{2q}}
    \|d_n\|_{L^{4q}}^2
    =
    O(\sqrt n\,r_n^2).
\end{aligned}
\]
Therefore
\[
    \sqrt n\,
    \|A_n^{(1)}+A_n^{(2)}+A_n^{(3)}\|_{L^q}
    =
    O(r_n+\sqrt n\,r_n^2)
    =
    o(1).
\]
This proves the expansion for $\bar\tha_n$.

It remains to show that the fixed projection does not alter the expansion.
Choose $r_\Theta>0$ such that
$B(\tha,r_\Theta)\subset\Theta$. Uniformly over $\|\bm h\|\le H$,
$\tha_n\in B(\tha,r_\Theta/2)$ for all large $n$. The preceding
expansion, the update-score moment bounds, and
$r_n=o(n^{-1/4})$ imply, for every $s<\infty$,
\[
    \|\bar\tha_n-\tha_n\|_{L^s}=O(n^{-1/2}).
\]
Hence, by taking arbitrarily high moments,
\[
    \Pr_{\tha_n}(\bar\tha_n\notin\Theta)
    \le
    \Pr_{\tha_n}
    \bigl(
        \|\bar\tha_n-\tha_n\|>r_\Theta/2
    \bigr)
\]
decays faster than any prescribed polynomial order.

Let
$\delta_n^\Pi=\Pi_\Theta(\bar\tha_n)-\bar\tha_n$. It vanishes whenever
$\bar\tha_n\in\Theta$. Since $\tha_n\in\Theta$, nearest-point projection
also gives
\[
    \|\delta_n^\Pi\|
    =
    \operatorname{dist}(\bar\tha_n,\Theta)
    \le
    \|\bar\tha_n-\tha_n\|.
\]
H\"older's inequality, the all-order moment bound, and the preceding exit
probability show
\[
    \sqrt n\,\|\delta_n^\Pi\|_{L^q}=O(n^{-M})
\]
for every $q<\infty$ and every prescribed $M>0$, after choosing a
sufficiently high auxiliary moment. By the rate convention above,
$r_n\ge n^{-1}$. Taking $M>1$ gives
$n^{-M}=O(r_n)$. Thus the projection displacement is absorbed into
$O_{L^q}(r_n+\sqrt n\,r_n^2)$ for every admissible localization-rate
envelope. Adding $\delta_n^\Pi$ proves
\eqref{eq:onestep}, including the stated all-order remainder.
\end{proof}

\emph{Proof of Lemma~\ref{lem:onestep}.}
Lemma~\ref{lem:onestepexp} applies to any pilot satisfying
\eqref{eq:pilotP} with $r_n=o(n^{-1/4})$, and Lemma~\ref{lem:pilot}
verifies this condition for the least-squares pilot. Together they supply
the full-parameter expansion. It remains only to identify its
loss-relevant block and record the moments.

Partition the Fisher matrix according to
$\tha=(\thL,\gamma)$. The efficient score and information are
\[
    \dot\ell_L
    =
    \dot\ell_{(\phi\beta)}
    -
    \Jb_{(\phi\beta)\gamma}
    \Jb_{\gamma\gamma}\inv
    \dot\ell_\gamma,
    \qquad
    \operatorname{Cov}(\dot\ell_L)=\Jb_L.
\]
The block-inverse identity gives, for every score vector,
\[
    [\Jb\inv\dot\ell]_L
    =
    \Jb_L\inv\dot\ell_L.
\]
Since
\[
    \sqrt n\,S_n(\tha_n)
    =
    \sqrt{\frac n{n_2}}\,
    \mathbb G_n,
    \qquad
    \mathbb G_n
    =
    n_2^{-1/2}
    \sum_{\ell>m_n}\dot\ell(\x_\ell;\tha_n),
\]
taking the $\thL$ block in \eqref{eq:onestep} yields
\[
    \sqrt n(\hat\thL_n-\thL_n)
    =
    \sqrt{\frac n{n_2}}\,
    \Jb_L(\tha_n)\inv\mathbb G_{L,n}
    +
    \xi_n,
\]
where
\[
    \mathbb G_{L,n}
    =
    n_2^{-1/2}
    \sum_{\ell>m_n}\dot\ell_L(\x_\ell;\tha_n)
\]
and $\xi_n=o_{L^q}(1)$ for every $q<\infty$.

The leading term has covariance
\[
    \frac n{n_2}\Jb_L(\tha_n)\inv
    \longrightarrow
    \Jb_L(\tha)\inv.
\]
Rosenthal's inequality gives
$\sup_n\sup_{\|\bm h\|\le H}
\E_{\tha_n}\|\mathbb G_{L,n}\|^p<\infty$ for every $p<\infty$.
Together with the all-order remainder bound and the uniform inverse bound,
this proves
\[
    \sup_n\sup_{\|\bm h\|\le H}
    \E_{\tha_n}
    \|\sqrt n(\hat\thL_n-\thL_n)\|^p
    <\infty.
\]

For the least-squares pilot,
$m_n=\min\{\lceil n^\omega\rceil,n-1\}$ with $\tfrac12<\omega<1$ and
$r_n=m_n^{-1/2}$. Hence
\[
    r_n+\sqrt n\,r_n^2
    =
    O(n^{-\omega/2}+n^{1/2-\omega})
    =
    O(n^{1/2-\omega})
    =
    o(1).
\]
Because $m_n=o(n)$, the pilot consumes a vanishing fraction of the
snapshots and the update block retains the full first-order Fisher
information. This proves Lemma~\ref{lem:onestep}. \hfill$\square$

\section{Proofs of Theorem~\ref{thm:main} and Corollary~\ref{cor:dist}}\label{app:thm}

The converse in Theorem~\ref{thm:main} follows directly from Theorem~\ref{thm:general}. Lemma~\ref{lem:reg} verifies quadratic-mean differentiability, Fisher nonsingularity, and differentiability of the MVDR functional at every regular ULA scene. Lemma~\ref{lem:compactreg} supplies the compact-set bounds used by the attaining construction, without quantitative separation. The remainder of this appendix proves the achievability and the limiting distribution.

\subsection{Achievability}

Fix $H<\infty$ and $\|\bm h\|\le H$, and abbreviate
\[
\begin{gathered}
    \tha_n=\tha+\frac{\bm h}{\sqrt n},
    \qquad
    \thL_n=(\tha_n)_L,\\
    \w_n^\star=\psi(\tha_n),
    \qquad
    \mathbf W_n=\mathbf W(\tha_n).
\end{gathered}
\]
The estimator is exactly the fixed-chart split one-step estimator
\eqref{eq:pilotdef}--\eqref{eq:onestepdef}. Put
\[
    U_n
    =
    \sqrt n(\hat\thL_n-\thL_n).
\]
Lemma~\ref{lem:onestep}, with $n_2=n-m_n$ and $m_n=o(n)$, gives
\begin{equation*}
    U_n
    =
    \sqrt{\frac n{n_2}}\,
    \Jb_L(\tha_n)\inv\mathbb G_{L,n}
    +
    \xi_n,
    \qquad
    \|\xi_n\|_{L^q}=o(1)
\end{equation*}
for every finite $q$, uniformly over $\|\bm h\|\le H$. Moreover,
\[
    \E_{\tha_n}\mathbb G_{L,n}=\bm0,
    \qquad
    \operatorname{Cov}_{\tha_n}(\mathbb G_{L,n})
    =
    \Jb_L(\tha_n),
\]
and $U_n$ has moments of every order uniformly on the local ball.

The map $\psi$ is twice continuously differentiable on a neighborhood of
the compact fixed chart and depends only on $\thL$. Taylor expansion
along the segment from $\tha_n$ to $\hat\tha_n$ gives
\begin{equation*}
\begin{aligned}
\sqrt n(\hat\w_n-\w_n^\star)
&=\nabla\psi(\tha_n)U_n+r_n^{(\psi)},\\
\|r_n^{(\psi)}\|
&\le\frac{C}{\sqrt n}\|U_n\|^2.
\end{aligned}
\end{equation*}
The segment lies in $\Theta$ because $\tha_n,\hat\tha_n\in\Theta$ and
$\Theta$ is convex in the operational chart. The constant is uniform in
the local ball because the second derivative of $\psi$ is bounded there.

Define
\[
    \ell_n
    =
    \|\hat\w_n-\w_n^\star\|_{\mathbf W_n}^2.
\]
Substituting the preceding Taylor expansion gives the quadratic
term:
\begin{equation*}
    n\ell_n
    =
    U_n^\top\Mb(\tha_n)U_n+r_n^{(\ell)},
    \qquad
    \sup_{\|\bm h\|\le H}
    \E_{\tha_n}|r_n^{(\ell)}|
    \longrightarrow0.
\end{equation*}
Indeed, the cross term is bounded by
$Cn^{-1/2}\|U_n\|^3$, and the squared Taylor remainder by
$Cn^{-1}\|U_n\|^4$. The uniform moments from
Lemma~\ref{lem:onestep} make both expectations vanish.

Let
\[
    V_n
    =
    \sqrt{\frac n{n_2}}\,
    \Jb_L(\tha_n)\inv\mathbb G_{L,n}.
\]
Then
\[
    \operatorname{Cov}_{\tha_n}(V_n)
    =
    \frac n{n_2}\Jb_L(\tha_n)\inv.
\]
Since $\xi_n=o_{L^2}(1)$ and $V_n$ is uniformly bounded in $L^2$,
\[
    \E_{\tha_n}
    \left[
    U_n^\top\Mb(\tha_n)U_n
    \right]
    =
    \frac n{n_2}
    \tr\!\left\{
    \Mb(\tha_n)\Jb_L(\tha_n)\inv
    \right\}
    +o(1),
\]
uniformly over $\|\bm h\|\le H$. Hence
\begin{equation*}
    \sup_{\|\bm h\|\le H}
    \left|
    \E_{\tha_n}[n\ell_n]
    -
    \frac n{n_2}
    \tr\!\left\{
    \Mb(\tha_n)\Jb_L(\tha_n)\inv
    \right\}
    \right|
    \longrightarrow0.
\end{equation*}

The plug-in MVDR weight is distortionless, so
Lemma~\ref{lem:excess} applies without an exceptional event:
\[
    L_n=\frac{\ell_n}{1+\ell_n}.
\]
Furthermore, $\psi$ is Lipschitz on $\Theta$, and the metrics
$\mathbf W_n$ are uniformly bounded. Therefore
\[
    n\ell_n\le C\|U_n\|^2
\]
pointwise. The fourth-moment bound for $U_n$ gives
\[
    \sup_{n\ge n_0}\sup_{\|\bm h\|\le H}
    \E_{\tha_n}(n\ell_n)^2<\infty.
\]
The exact difference between the quadratic loss and the SINR loss satisfies
\[
    0
    \le
    n\ell_n-nL_n
    =
    \frac{n\ell_n^2}{1+\ell_n}
    \le
    \frac{(n\ell_n)^2}{n},
\]
and consequently
\[
    \sup_{\|\bm h\|\le H}
    \left|
    \E_{\tha_n}[nL_n]
    -
    \E_{\tha_n}[n\ell_n]
    \right|
    \longrightarrow0.
\]

Finally, $n/n_2\to1$, while continuity of $\Mb$ and $\Jb_L^{-1}$ gives
\[
    \sup_{\|\bm h\|\le H}
    \left|
    \tr\!\left\{
    \Mb(\tha_n)\Jb_L(\tha_n)\inv
    \right\}
    -
    \tr(\Mb(\tha)\Jb_L(\tha)\inv)
    \right|
    \longrightarrow0.
\]
Combining the last four displays proves
\[
    \sup_{\|\bm h\|\le H}
    \left|
    n\,\E_{\tha+\bm h/\sqrt n}[L(\hat\w_n)]
    -\tr(\Mb\Jb_L\inv)
    \right|
    \longrightarrow0,
\]
which is the achievability assertion of Theorem~\ref{thm:main}.
\hfill$\square$

\subsection{Proof of Corollary~\ref{cor:dist}}

Let $\bm h_n$ be any sequence with $\|\bm h_n\|\le H$ and put
$\tha_n=\tha+\bm h_n/\sqrt n$. The one-step expansion gives
\[
    U_n
    =
    \sqrt n(\hat\thL_n-\thL_n)
    =
    \sqrt{\frac n{n_2}}\,
    \Jb_L(\tha_n)\inv\mathbb G_{L,n}
    +o_{L^q}(1)
\]
for every finite $q$. The update-block summands in
$\mathbb G_{L,n}$ are independent and centered. Their moments of every
order are uniformly bounded, and
$\Jb_L(\tha_n)\to\Jb_L(\tha)$. The Lyapunov condition therefore gives
\[
    U_n
    \Rightarrow
    U\sim
    \mathcal N
    (\bm0,\Jb_L(\tha)\inv).
\]

The achievability calculation already established
\[
    n\ell_n
    =
    U_n^\top\Mb(\tha_n)U_n+o_{L^1}(1),
    \qquad
    nL_n-n\ell_n\longrightarrow0
    \quad\text{in }L^1.
\]
Since $\Mb(\tha_n)\to\Mb(\tha)$, the continuous-mapping theorem gives
\[
    nL_n
    \Rightarrow
    U^\top\Mb(\tha)U.
\]
Diagonalize
\[
    \Jb_L^{-1/2}\Mb\Jb_L^{-1/2}
    =
    \mathbf Q\,
    \operatorname{diag}(\lambda_1,\ldots,\lambda_{2K})
    \mathbf Q^\top.
\]
If
$\mathbf Z=\mathbf Q^\top\Jb_L^{1/2}U$, then
$\mathbf Z\sim\mathcal N(\bm0,\Ib)$, and hence
\[
    U^\top\Mb U
    \ \stackrel{d}{=}\
    \sum_{s=1}^{2K}\lambda_s Z_s^2,
\]
where $\stackrel{d}{=}$ denotes equality in distribution.

It remains to pass all moments. For every finite $p>0$,
Lemma~\ref{lem:onestep} gives a uniform $4p$-moment bound for $U_n$, and
the pointwise estimate $nL_n\le n\ell_n\le C\|U_n\|^2$ gives
\[
    \sup_n\E_{\tha_n}(nL_n)^{2p}<\infty.
\]
Thus $(nL_n)^p$ is uniformly integrable. Distributional convergence
therefore implies convergence of every moment to the corresponding
moment of $\sum_s\lambda_s Z_s^2$, proving
Corollary~\ref{cor:dist}. \hfill$\square$

\section{Proof of Theorem~\ref{thm:smi}}\label{app:smi}

The proof has three components. Whitening gives the exact covariance-independent SMI risk. The unrestricted Gaussian experiment is then inserted into Theorem~\ref{thm:general}. Finally, an explicit covariance--weight derivative evaluates the resulting information trace as $N-1$.

\subsection{The exact SMI risk}

The Reed--Mallett--Brennan law
\cite{ReedMallettBrennan1974} states that, for signal-free
complex-Gaussian training with $N\ge2$ and $n\ge N$,
\[
    \rho_{\mathrm{SMI}}
    \sim
    \operatorname{Beta}(n-N+2,N-1).
\]
Therefore
\[
    \E[L_{\mathrm{SMI}}]
    =
    1-\E[\rho_{\mathrm{SMI}}]
    =
    \frac{N-1}{n+1}.
\]

To record why the law is independent of the covariance, set
\[
    \widehat\Rb
    =
    \frac1n\sum_{\ell=1}^n\x_\ell\x_\ell\herm,
    \qquad
    \widehat\w_{\mathrm{SMI}}
    =
    \frac{\widehat\Rb\inv\ba_0}
    {\ba_0\herm\widehat\Rb\inv\ba_0}.
\]
Under the whitening transformation
\[
\begin{gathered}
    \mathbf y_\ell=\Rb^{-1/2}\x_\ell,
    \qquad
    \widehat{\mathbf S}
    =
    \Rb^{-1/2}\widehat\Rb\Rb^{-1/2},\\
    \mathbf b_0=\Rb^{-1/2}\ba_0,
\end{gathered}
\]
the output-SINR ratio becomes the same ratio computed from
$(\widehat{\mathbf S},\mathbf b_0)$ in white noise. The ratio is
unchanged by rescaling $\mathbf b_0$ and by a common unitary rotation.
Rotating $\mathbf b_0/\|\mathbf b_0\|$ to $\mathbf e_1$ reduces every
$\Rb\succ0$ to the same white experiment. Hence
\begin{equation*}
    n\,\E_\Rb[L_{\mathrm{SMI}}]
    =
    \frac{n(N-1)}{n+1}
\end{equation*}
at every covariance, including every local alternative.

\subsection{The unrestricted Gaussian experiment}

Let $\{\mathbf H_a\}_{a=1}^{N^2}$ be a fixed real basis of the Hermitian
matrices and parametrize the open cone by
\[
    \Rb(\bm\vartheta)
    =
    \sum_{a=1}^{N^2}\vartheta_a\mathbf H_a
    \succ0.
\]
The complex-Gaussian density is smooth on this cone. The argument of
Lemma~\ref{lem:reg}(ii), which uses only smoothness and positive
definiteness, gives quadratic-mean differentiability with score
\[
    \dot\ell_a(\x;\bm\vartheta)
    =
    \tr\!\left[
    \Rb\inv\mathbf H_a\Rb\inv
    (\x\x\herm-\Rb)
    \right]
\]
and Fisher information
\begin{equation*}
    (\Jb_{\mathrm{unstr}})_{ab}
    =
    \tr(
    \Rb\inv\mathbf H_a
    \Rb\inv\mathbf H_b).
\end{equation*}
For a nonzero Hermitian
$\mathbf H=\sum_a u_a\mathbf H_a$,
\[
    \bm u^\top\Jb_{\mathrm{unstr}}\bm u
    =
    \|\Rb^{-1/2}\mathbf H\Rb^{-1/2}\|_F^2
    >0,
\]
so the information is nonsingular.

The covariance--information identity used below can be verified directly
in this chart. Let $\{\mathbf H^a\}$ be the trace-dual basis,
$\tr(\mathbf H^a\mathbf H_b)=\delta_{ab}$, and define
\[
    Y_a=\tr(\mathbf H^a\x\x\herm).
\]
Then $\E_{\bm\vartheta}Y_a=\vartheta_a$, and
\[
    \x\x\herm-\Rb
    =
    \sum_{b=1}^{N^2}(Y_b-\vartheta_b)\mathbf H_b.
\]
Substitution into the score gives
\[
    \dot{\bm\ell}
    =
    \Jb_{\mathrm{unstr}}
    (\mathbf Y-\bm\vartheta).
\]
Taking covariance and using
$\operatorname{Cov}(\dot{\bm\ell})=\Jb_{\mathrm{unstr}}$ yields
\[
    \operatorname{Cov}(\mathbf Y)
    =
    \Jb_{\mathrm{unstr}}\inv.
\]
Since the coordinate vector of the sample covariance is
$\widehat{\bm\vartheta}=n^{-1}\sum_{\ell=1}^n\mathbf Y_\ell$,
\begin{equation*}
    n\,\operatorname{Cov}(\widehat{\bm\vartheta})
    =
    \Jb_{\mathrm{unstr}}\inv.
\end{equation*}

We also record the derivative that connects covariance estimation to
MVDR-weight estimation. On the full positive-definite cone, let
\[
    \zeta=\ba_0\herm\Rb\inv\ba_0=P_\star^{-1},
    \qquad
    \psi(\Rb)=\frac{\Rb\inv\ba_0}{\zeta}.
\]
For a Hermitian direction $\mathbf H$, differentiation of the inverse and
the normalizing denominator gives
\begin{equation*}
    D\psi_\Rb[\mathbf H]
    =
    -\Rb\inv\mathbf H\,\wopt
    +
    \zeta\,(\wopt\herm\mathbf H\wopt)\wopt.
\end{equation*}
This formula uses no source structure and also shows
$\ba_0\herm D\psi_\Rb[\mathbf H]=0$. In the fixed real chart, define
\[
    (\Mb_{\mathrm{full}})_{ab}
    =
    \frac{
    \operatorname{Re}
    \langle
    D\psi_\Rb[\mathbf H_a],
    D\psi_\Rb[\mathbf H_b]
    \rangle_\Rb
    }{P_\star}.
\]
This is precisely the quadratic form of the linearized normalized excess
loss. Both $\Mb_{\mathrm{full}}$ and
$\Jb_{\mathrm{unstr}}$ transform by congruence under a change of real
chart. Hence,
$\tr(\Mb_{\mathrm{full}}\Jb_{\mathrm{unstr}}\inv)$ is chart-independent.

\subsection{Application of the general converse}

Fix an interior covariance $\Rb(\bm\vartheta)\succ0$. The preceding calculation shows that the unrestricted complex-Gaussian experiment is quadratic-mean differentiable with nonsingular information $\Jb_{\mathrm{unstr}}$. The derivative in (\ref{eq:full-derivative}) shows that the MVDR functional is differentiable on the positive-definite cone. Hence Theorem~\ref{thm:general} applies and gives
\begin{equation*}
\lim_{r\to\infty}\liminf_n\inf_{\hat\w_n}
\sup_{\|\bm h\|\le r}
\E_{\bm\vartheta+\bm h/\sqrt n}[nL(\hat\w_n)]
\ge\tr(\Mb_{\mathrm{full}}\Jb_{\mathrm{unstr}}^{-1}).
\end{equation*}
The next subsection evaluates this trace.

\subsection{Evaluation of the full trace}

The whitening and unitary transformations used above are invertible
equivalences of the experiment and preserve the SINR loss. Together
with chart invariance of the trace, they reduce its evaluation to
\[
    \Rb=\Ib,
    \qquad
    \ba_0=\mathbf e_1,
    \qquad
    P_\star=1.
\]
At this point, the MVDR derivative becomes
\begin{equation*}
    D\psi_{\Ib}[\mathbf H]
    =
    -\mathbf H\mathbf e_1
    +
    H_{11}\mathbf e_1
    =
    -\mathbf P_{\mathbf e_1}^{\perp}
    \mathbf H\mathbf e_1.
\end{equation*}
Thus the linearized excess is
\[
    E_{\mathrm{lin}}(\mathbf H)
    =
    \|D\psi_{\Ib}[\mathbf H]\|^2
    =
    \sum_{i=2}^N|H_{i1}|^2.
\]

Take $\mathbf H=\widehat\Rb-\Ib$. The exact covariance--information
identity gives
\begin{equation*}
\begin{aligned}
    \tr(\Mb_{\mathrm{full}}\Jb_{\mathrm{unstr}}\inv)
    &=
    n\,\E\big[
    E_{\mathrm{lin}}(\widehat\Rb-\Ib)
    \big]\\
    &=
    \sum_{i=2}^N
    n\,\E|\widehat R_{i1}|^2.
\end{aligned}
\end{equation*}
For $i\ne1$,
\[
    \widehat R_{i1}
    =
    \frac1n\sum_{\ell=1}^n
    x_{\ell i}\overline{x_{\ell1}},
    \qquad
    \E|\widehat R_{i1}|^2=\frac1n
\]
at the white point. Hence
\begin{equation*}
    \tr(\Mb_{\mathrm{full}}\Jb_{\mathrm{unstr}}\inv)
    =
    \sum_{i=2}^N1
    =
    N-1.
\end{equation*}

The same coefficient can be read directly from the Hermitian tangent
geometry. For each $j=2,\ldots,N$, the two Frobenius-orthonormal directions
\[
\frac{\mathbf e_j\mathbf e_1^*+\mathbf e_1\mathbf e_j^*}{\sqrt2},
\qquad
\frac{i(\mathbf e_j\mathbf e_1^*-\mathbf e_1\mathbf e_j^*)}{\sqrt2}
\]
each have unit Fisher information and task curvature $1/2$. All other
elements of the standard Hermitian basis have zero task curvature.
Thus the $2(N-1)$ loss-relevant real directions contribute
$2(N-1)\times(1/2)=N-1$, with the same factor $1/2$ as in
(\ref{eq:hermitian-task-bound}).

\subsection{Achievability by SMI}

For every fixed $r<\infty$ and all sufficiently large $n$, the exact covariance-independent risk gives
\[
    \sup_{\|\bm h\|\le r}
    n\,\E_{\bm\vartheta+\bm h/\sqrt n}
    [L_{\mathrm{SMI}}]
    =
    \frac{n(N-1)}{n+1}.
\]
Taking first $n\to\infty$ and then $r\to\infty$ shows that SMI matches the converse in the same iterated expanding-neighborhood sense. The limit $N-1$ equals the unrestricted converse constant
$\tr(\Mb_{\mathrm{full}}\Jb_{\mathrm{unstr}}\inv)$. Therefore, the converse and achievability in Theorem~\ref{thm:smi} meet at the same coefficient. \hfill$\square$

\bibliographystyle{IEEEtran}
\bibliography{references}

\end{document}